%% file: main.tex
\documentclass[journal,twoside,web]{ieeecolor}
\usepackage{generic}
\usepackage{amsmath,amsfonts,amssymb}  
\usepackage{mathtools}
\usepackage{graphicx}
\usepackage{epsfig}
\usepackage{accents}
\usepackage{booktabs}

\usepackage{cite}
\usepackage{hyperref}
\hypersetup{hidelinks=true}
\usepackage{textcomp}
\usepackage{balance}
\usepackage{picinpar}
\usepackage{url}
\usepackage{flushend}
\usepackage{colortbl}
\usepackage{placeins}
\usepackage{xcolor}
\usepackage{svg}
\usepackage{soul}
\usepackage{multirow}
\usepackage{tikz}
\usetikzlibrary{shapes,shapes.geometric,arrows,fit,calc,positioning,automata}
\usepackage{verbatim}
\usepackage{pifont}
\usepackage{alltt}
\usepackage{enumerate}
\usepackage{siunitx}
\usepackage{breakurl}
\usepackage{epstopdf}
\usepackage{pbox}
\usepackage{float}
\usepackage{lipsum}
\let\IEEEproof\proof
\let\endIEEEproof\endproof
\let\proof\relax
\let\endproof\relax

\usepackage{amsthm}
\let\proof\IEEEproof
\let\endproof\endIEEEproof
\theoremstyle{plain}
\newtheorem{thm}{Theorem}
\newtheorem{lem}{Lemma}
\newtheorem{cor}{Corollary}

\theoremstyle{definition}
\newtheorem{assum}{Assumption}

\theoremstyle{remark}
\newtheorem{rem}{Remark}
\usepackage{orcidlink}
\def\BibTeX{{\rm B\kern-.05em{\sc i\kern-.025em b}\kern-.08em
    T\kern-.1667em\lower.7ex\hbox{E}\kern-.125emX}}
\begin{document}
\title{Output Feedback Adaptive Performance Control}
\author{Panagiotis S. Trakas$^{\orcidlink{0000-0002-6064-7370}}$ and Charalampos P. Bechlioulis$^{\orcidlink{0000-0001-9850-2540}}$, \IEEEmembership{Senior Member, IEEE}
\thanks{This work was supported by the
Hellenic Foundation for Research and Innovation (H.F.R.I.) under the second call for research projects to support post-doctoral researchers (HFRI-PD19-370).}
\thanks{P. S. Trakas is with the Department of Information Technology and Electrical Engineering, ETH Zurich, Switzerland. C. P. Bechlioulis is with the Department of Electrical and Computer Engineering, University of Patras, Greece. E-mails: {\tt\small ptrakas@control.ee.ethz.ch, chmpechl@upatras.gr.}}
}
\maketitle
\begin{abstract}
In this paper, we consider uncertain high-order nonlinear systems performing dynamic tracking tasks under hard actuator constraints, where only the output error is available for measurement, while the system states and the desired trajectory derivatives are unavailable for feedback. We propose a robust output-feedback controller that guarantees adaptive performance specifications in this framework. The proposed scheme employs a novel Prescribed Performance Observer (PPO) with dynamic gains, which enhances estimation accuracy while avoiding large fixed observer gains. In addition, we introduce an adaptive mechanism that dynamically adjusts the output performance specifications according to the actuator limitations, ensuring bounded closed-loop signals. We establish a separation principle showing that the output-feedback scheme recovers the performance of its state-feedback counterpart. Comparative simulations demonstrate accurate tracking and
smoother applied control under actuator limitations,
uncertainties, and measurement noise.
\end{abstract}

\begin{IEEEkeywords}
Prescribed performance control; Nonlinear output regulation; Robust control; Input constraints; Feedback constraints. 
\end{IEEEkeywords}

\section{Introduction}
\IEEEPARstart{T}{he} control of nonlinear systems in the presence of input constraints and limited feedback remains a fundamental challenge in modern control theory. Actuator limitations, inherent in every physical system, often lead to degraded performance or instability, especially in uncertain nonlinear systems~\cite{survey}. At the same time, feedback limitations, where only system outputs are measurable, aggravate these issues, necessitating the use of state estimators to reconstruct unmeasured states~\cite{hgo8}. Moreover, certain applications concern dynamic tasks where the desired trajectory is not a priori known for all time, as an explicit function of time, but is continuously measured. Such scenarios involve tracking of moving targets, e.g., in autonomous vehicles or robotic systems, where the desired trajectory, i.e., the target position, is obtained at each time instant via measuring devices. As this trajectory is unknown beforehand, its derivatives cannot be exploited for feedback, and thus the high-order derivatives of the tracking error are unavailable. These practical constraints combined with the requirement for a closed-loop system that meets certain user-defined performance specifications regarding both transient and steady-state response, make the control design particularly challenging.

High-gain observers (HGOs) \cite{hgo8,atasi,hgo} are a cornerstone of output-feedback control for nonlinear systems, offering robustness to model uncertainties and achieving closed-loop performance comparable to state-feedback designs when observer gains are sufficiently high. However, HGOs amplify measurement noise, particularly in steady state, and are prone to the peaking phenomenon during transient, which can destabilize the closed-loop system (see, e.g., \cite{khalil, prasov, freidovich, ran}). To address these issues, various advancements to the standard HGO have been proposed \cite{Tarbouriech}. Adaptive HGOs \cite{c11,ahrens} dynamically adjust observer gains to balance noise sensitivity and estimation accuracy, while filtered HGOs \cite{fhgo,mousavi} introduce mechanisms to attenuate noise amplification; however, they often incur computational overhead, making real-time implementation challenging. Nonlinear output injection strategies \cite{astolfiA,ball} further refine the observer response to noise and disturbances. In \cite{astolfi}, a HGO with a reduced gain order of power 2 (compared to the standard \( n \)) was introduced, improving noise rejection, but at the expense of increasing the observer state dimension from \( n \) to \( 2n-2 \). Although this approach mitigates noise amplification, the peaking phenomenon remains unresolved. An alternative solution proposed in \cite{khalil}, involves cascading lower-dimensional observers with saturation functions applied between them. The cascade HGO reduces the observer gain order to \( k \) (from the standard \( k^n \)), offering a numerical advantage over conventional HGOs. However, as it was shown in \cite{khalil}, it inherits the main issues of the standard HGO and demonstrates similar performance attributes and robustness to measurement noise in feedback control applications.

While the aforementioned approaches improve robustness to noise or reduce peaking, they often require introducing additional parameters, increasing the complexity of the observer design and its implementation in an output-feedback control framework. Moreover, such methods often assume linear or constant input gains, which do not reflect the nonlinear and state-dependent characteristics encountered in systems like robot manipulators or aerial vehicles \cite{guo, tac23}. Extended state observers (ESOs) provide an alternative by simultaneously estimating system states and lumped uncertainty, including external forces and unmodeled dynamics \cite{eso}. These observers have found success in active disturbance rejection control (ADRC) in systems with partial state measurements \cite{ran}. Nonetheless, ESO designs also face challenges with peaking, sensitivity to noise, and reliance on a matching condition between the nominal and actual control gains \cite{freidovich, krstic, tac24}. Although methods like cascading ESOs with saturation mechanisms have been proposed to manage internal variables, they often depend on heuristic selection of observer gains, lacking a versatile framework for employing distinct nonlinear gains, while ensuring convergence and preserving performance guarantees \cite{c12}.
\par Additionally, imposing transient and steady-state performance specifications is essential in practical control systems. However, the related literature on output-feedback control remains in premature stages. Approximation-free methodologies, such as prescribed performance control (PPC) \cite{bechlaut} and funnel control (FC) \cite{funnelc}, have been developed to ensure that the tracking error remains within a predefined performance funnel, satisfying user-defined output performance requirements (see also notable works such as \cite{bechltac,bergerF,bergerIC,bechlaut9}). The authors in \cite{berger}, introduced a
dynamic output-feedback funnel controller for
minimum phase systems with a relative degree of two. While funnel-based output feedback methods \cite{funnel, dimanidis} provide semiglobal solutions to the output-feedback tracking problem with predefined transient and steady-state bounds, they typically involve extensive gain tuning and fail to explicitly address issues such as saturation effects and estimation errors. In particular, input saturation limits are selected heuristically to mitigate the peaking phenomenon associated with HGOs, without adequately accounting for the hard input constraints imposed by actuator limitations. However, the coupling between the observer and plant dynamics can result in significant performance degradation, leading to tracking errors that exceed prescribed bounds \cite{dimanidis,funnel}. This occurs both during transient, due to the peaking phenomenon, and in the steady-state, as a result of noise amplification inherent in observer-based approaches. An adaptive output-feedback funnel controller for uncertain nonlinear systems was recently proposed in \cite{c5}, achieving output tracking with prescribed transient behavior. While this controller provides global results, it relies on assumptions such as growth bounds on the system's drift term and constant control input gain. Additionally, the approach focuses on a specific class of systems where the nonlinear terms depend solely on the measured system output, thereby neglecting the couplings between unknown nonlinearities and the unmeasured states, which constitutes a challenging aspect in output feedback control.

Achieving a unified solution  that combines robustness, high-performance control, and rigorous handling of actuator constraints for uncertain nonlinear systems with feedback limitations, remains an open problem. To this end, we propose a robust output-feedback control scheme for input-constrained nonlinear systems ensuring flexible performance specifications building on our previous work on the adaptive performance control (APC) \cite{cdc22,smc,ratesiso}. APC has emerged as an extension of the PPC method, for addressing both input and output constraints in uncertain nonlinear systems. Within the APC framework the soft performance constraints are dynamically adjusted in real time to address the strict input limitations imposed by the actuators of the system, ensuring a feasible control input while achieving the best feasible performance respecting these constraints. In this paper, we extend APC by designing a control scheme that dynamically adjusts the convergence rate of the output error based on the available control effort. This mechanism enhances system performance during multiple transient phases that may arise due to actuator saturation, ensuring exponential convergence while taking into account the hard input constraints, achieving the best feasible performance specifications within the physical limitations of the actuators. Furthermore, we design a novel nonlinear observer, named prescribed performance observer (PPO), which extends the standard High-Gain Observer (HGO) by incorporating time and error dependent dynamic gains. Notably, the PPO retains the computational simplicity of standard HGO. Unlike conventional HGOs with fixed gain coefficients, the
PPO adjusts its injection coefficients dynamically according
to time and the output-estimation error. Although the
time-dependent coefficients decrease exponentially fast, the transformed injection results in a steady-state differential gain of order $O(n)$, as in a classical HGO. Finally, the PPO is integrated into a robust APC control architecture that explicitly accounts for input constraints, including amplitude and rate limitations of the control signal. For clarity and ease of presentation, the theoretical analysis is conducted for uncertain SISO nonlinear systems in the normal form. Nevertheless, the proposed methodology can be readily extended to feedback linearizable MIMO systems that can be transformed into the Byrnes-Isidori normal form \cite{isidorib}, broadening its applicability to more complex control scenarios. The main contributions of this work are summarized as follows:
\begin{itemize}
    \item We introduce the PPO, a novel nonlinear observer with
time and error-dependent injection gains that provides fast
transient error correction without requiring plant-model
information, while retaining the computational simplicity of
a standard HGO \cite{khalil,freidovich}.
     \item  Contrary to the related literature on APC \cite{smc,ratesiso,bergerIC,cdc22,acc}, we propose an input-dependent convergence rate function, that enables the system to recover the prescribed performance specifications while explicitly leveraging the available control effort during the transient period, thereby accelerating convergence to the predefined steady-state set under hard input constraints.
     \item We develop a unified output-feedback framework that integrates the PPO with the APC method under simultaneous hard amplitude and rate constraints. Unlike our earlier state-feedback APC designs \cite{cdc22,smc,ratesiso}, and unlike funnel-based output-feedback controllers with virtual (tunable) input saturation limits \cite{funnel,dimanidis,acc}, the proposed scheme provides feasibility conditions that guarantee bounded closed-loop signals as well as prescribed performance attributes under hard input constraints and feedback limitations.
\end{itemize}
The remainder of the paper is structured as follows. Section \ref{prel} defines the problem addressed in this work and presents the necessary preliminary knowledge. Section~\ref{stfeed} develops the
state-feedback controller, while Section~\ref{PPO} introduces
the PPO. Section~\ref{main} integrates the PPO with the APC
scheme and establishes recovery of the state-feedback
performance under output feedback. Section~\ref{simsec}
presents the simulation results, and
Section~\ref{conclusio} concludes the paper.
\section{Problem statement and Preliminaries}\label{prel}
We consider a nonlinear system in the normal form \cite{isidorib}:
\begin{equation}\label{eq:system}
\begin{split}
\dot{z} &= f_0(x, z, d) \\
\dot{x} &= Ax + B[f(x, z, d) + g(x,z,d)u] \\
y &= C^Tx 
\end{split}
\end{equation}
\noindent where $z \in \mathbb{R}^m$ and $x \in \mathbb{R}^n$ are the system states, $y \in \mathbb{R}$ is the measured output and $d \in \mathbb{R}^{\iota}$ is a piecewise continuous and bounded term denoting disturbances. The control input \( u \) is subject to both amplitude and rate constraints imposed by the actuator capabilities. Specifically, the amplitude of \( u \) is restricted to the compact set \( \mathcal{U} \coloneqq [-\Bar{u}, \Bar{u}] \), where \( \Bar{u} > 0 \) is a known constant. Additionally, the rate of change of \( u \), denoted by \( \dot{u} \), is limited to the compact set \( \mathcal{R} \coloneqq [-\Bar{r}, \Bar{r}] \), with \( \Bar{r} > 0 \) as a known constant. The system nonlinearities $f_0:\mathbb{R}^n \times \mathbb{R}^m \times \mathbb{R}^\iota \rightarrow \mathbb{R}^m$ and $f,g: \mathbb{R}^n \times \mathbb{R}^m \times \mathbb{R}^\iota \rightarrow \mathbb{R}$ are unknown, locally Lipschitz nonlinear functions. The matrices $A \in \mathbb{R}^{n\times n},~B \in \mathbb{R}^{n}$ and $C \in \mathbb{R}^{n}$ are given by:
\begin{align*}
    A \coloneqq \begin{bmatrix}
    \textbf{0}_{n-1} & \textbf{\textit{I}}_{n-1} \\ 0 & \textbf{0}_{n-1}^T
\end{bmatrix},~ B = \begin{bmatrix}
    \textbf{0}_{n-1} \\ 1
\end{bmatrix},~C \coloneqq \begin{bmatrix}
    1 \\ \textbf{0}_{n-1}
\end{bmatrix}
\end{align*}
with \(\textbf{\textit{I}}_{n-1}\) and \(\textbf{0}_{n-1}\) denoting the \((n-1) \times (n-1)\) identity matrix and a column vector of length \(n-1\) containing only zeros, respectively. Finally, consider a measurable reference trajectory $y_r(t) \in \mathbb{R}$ and let us define the vector $x_d(t) \coloneqq \begin{bmatrix} y_r(t), &\left(\frac{d}{dt}\right)y_r(t), & \cdots &, \left(\frac{d}{dt}\right)^{n-1}y_r(t) \end{bmatrix}^T$, containing the reference trajectory $y_r$ and its derivatives up to order $n-1$, and the state tracking error $e(t) \coloneqq x(t) - x_d(t)$.

The objective of this work is to design a dynamic output-feedback controller for system \eqref{eq:system} such that:
\begin{itemize}
    \item all closed-loop signals remain uniformly ultimately bounded;
    \item the tracking error \( e_1(t) = y(t) - y_r(t) \) meets adaptive performance specifications, according to the actuation and feedback limitations;
    \item in the absence of input saturation, $e_1(t)$ converges exponentially to a predefined small neighborhood of the origin with a prescribed minimum convergence rate.
\end{itemize} 
The following assumptions are made throughout the paper.
\begin{assum}\label{ass1}
   The sign of the unknown function $g(x,z,d)$ is known and there exists a positive constant $g^*$ such that $\inf_{(x,z,d) \in \mathbb{R}^{n+m+\iota}}\{|g(x,z,d)|\} = g^*$. Without loss of generality it is assumed that $g(x,z,d)$ is positive.
\end{assum}

\begin{assum}\label{ass2}
The initial plant state satisfies $(z(0),x(0))\in\mathcal Z_0\times\mathcal X_0$, where \(\mathcal Z_0\subset\mathbb R^m\) and
\(\mathcal X_0\subset\mathbb R^n\) are known compact
sets. Moreover, the internal dynamics of system
\eqref{eq:system} is bounded-input-bounded-state stable with
respect to \(x\) and \(d\). In particular, there exist a
continuously differentiable, radially unbounded, positive
definite function \(V_0:\mathbb R^m\to\mathbb R_{\geq0}\) and a
nonnegative continuous function
\(\varphi:\mathbb R^n\times\mathbb R^\iota\to\mathbb R_{\geq0}\)
such that
\[
\frac{\partial V_0}{\partial z}(z)f_0(x,z,d)\leq0
\quad\text{whenever}\quad
\|z\|\geq\varphi(x,d)
\]
for all \(x\in\mathbb R^n\), \(z\in\mathbb R^m\), and
\(d\in\mathbb R^\iota\).
\end{assum}
\begin{assum}\label{ass3}
    The reference trajectory \(y_r(t)\), is a known and bounded function of time as well as continuously differentiable up to order \(n\), with unknown bounded derivatives. Specifically, there exists an unknown compact set $\mathcal{Y}$ where the reference vector $x_d(t)$ strictly lies in, i.e., $x_d(t) \in \mathcal{Y} \subset \mathbb{R}^{n},~ \forall t \geq 0$.
\end{assum}
\begin{rem}\label{assumptions}
    Assumption \ref{ass1} guarantees the controllability of system \eqref{eq:system}, making the proposed method applicable to uncertain nonlinear systems with state-dependent control gains, provided that only the control direction is known. This is a notable distinction from the majority of existing literature \cite{c5,funnel,berger,ran}, which typically assumes either known or state-independent control input gain $g(\cdot)$. Furthermore, Assumption \ref{ass2} establishes that the system is minimum-phase, which is common in output-feedback control framework \cite{freidovich,c5,funnel,dimanidis}, as it decouples the internal dynamics from the tracking control problem. Finally, Assumption \ref{ass3}, states that only the reference trajectory $y_r(t)$ is directly measurable. This is particularly relevant for applications such as autonomous vehicles, where the reference trajectory (e.g., target position) is continuously determined in real-time. Relaxing Assumptions \ref{ass1} and \ref{ass2}, addressing unknown control directions and nonminimum-phase dynamics, is left open for future research.

\end{rem}
\subsection{Notations}\label{notations}
For clarity and consistency, we adopt the following notation and definitions throughout the paper.
\begin{itemize}
    \item \(\mathbb{R}\), \(\mathbb{R}_+\), \(\mathbb{R}^*_+\), \(\mathbb{N}\) denote the set of real numbers, non-negative real numbers, positive real numbers, and natural numbers, respectively.  
    \item  \(|\cdot|\), \(\|\cdot\|\), \(\|\cdot\|_\infty\) denote the absolute value (\(\mathcal{L}_1\) norm) of a scalar ( vector), the spectral (Euclidean) norm, and the infinity norm of a matrix (vector), respectively. 
    \item $\lambda_{\max}(A)$, $\lambda_{\min}(A)$ denote the maximum and minimum eigenvalue of the matrix $A$, respectively.
    \item A continuous function \(f: [0, a] \to \mathbb{R}_+\) belongs to class \(\mathcal{K}\) if \(f(0) = 0\) and \(f\) is strictly increasing. A function \(f\) belongs to class \(\mathcal{K}_\infty\) if \(f \in \mathcal{K}\) with \(a = \infty\) and \(\lim_{t \to \infty} f(t) = \infty\). \item For a differentiable function \( f(\chi) \), the prime symbol (\('\)) represents differentiation with respect to \(\chi\), i.e., \( f'(\chi) = \frac{d f(\chi)}{d\chi} \).
    \item We define a differentiable saturation function as:  
\begin{align*}
      \mathrm{sat}_{\sigma}(\chi) \coloneqq 
  \begin{cases} 
  \chi & \text{if } |\chi| < \sigma - \beta, \\
  p(\chi) & \text{if } |\chi| \in [\sigma - \beta, \sigma + \beta], \\
  \operatorname{sign}(\chi) \sigma & \text{if } |\chi| > \sigma + \beta
  \end{cases}
\end{align*}  
with \(p(\chi) =-\frac{\operatorname{sign}(\chi)}{4\beta}
\left(
|\chi|^2-2(\sigma+\beta)|\chi|+(\sigma-\beta)^2
\right)\), where \(\beta = 10^{-6}\) is a small smoothing parameter.
  \item We define the mapping \( T(\chi) \coloneqq \frac{1}{2}\ln{\left( \frac{1+\chi}{1- \chi} \right)} \) along with its Jacobian \( J(\chi) \coloneqq \frac{1}{ 1-\chi^2} \) and their product \( \mathcal{T}(\chi) \coloneqq J(\chi)T(\chi)= \frac{\ln{\left( \frac{1+\chi}{1- \chi} \right)}}{2 \left( 1-\chi^2 \right)} \), with \(T: (-1,1) \rightarrow \mathbb{R},J: (-1,1) \rightarrow [1,\infty), \mathcal{T}: (-1,1) \rightarrow \mathbb{R} \).
\end{itemize} 
\subsection{Preliminaries on Adaptive Performance Control}\label{APCsec}
APC focuses on ensuring that the output tracking error converges to a predefined, arbitrarily small residual set with a specified minimum convergence rate, provided that the input constraints allow it. In particular, according to the approach presented in \cite{smc}, given the system \eqref{eq:system} and a reference trajectory $y_r(t)$, we define the measurable tracking error $e_1(t) \coloneqq y(t)-y_r(t)$. Adaptive performance is then achieved if \( e_1(t) \) remains within the performance funnel defined by the bounds \( [-\rho(t), \rho(t)] \) for all \( t \geq 0 \). Note that the initial tracking error \( e_1(0) \) must lie within the performance envelope at \( t=0 \), i.e., \( \rho(0) > \lvert e(0) \rvert \). Subsequently, the adaptive performance law is given by:
\begin{equation} \label{convPF}
    \dot{\rho} = -l(\rho(t) - \rho_\infty) + \chi(x,x_d,\rho).
\end{equation}
 The first term in (\ref{convPF}) represents the evolution of the nominal exponential function, i.e. $(\rho(t)-\rho_{\infty})\exp{(-l t)} + \rho_{\infty}$, where the parameter $\rho_\infty>0$ denotes the maximum allowable steady-state value of $e_1(t)$ and $l>0$ determines its minimum convergence rate. The second term, $\chi(x,x_d,\rho)$ is a properly designed function that dynamically adjusts the predefined performance attributes when input saturation occurs. By properly amplifying $\rho(t)$, the non-negative term $\chi(x,x_d,\rho)$ ensures that the closed-loop signals remain bounded under input constraints. 
\section{Controller and Observer Design}
\subsection{State Feedback Controller}\label{stfeed}
In this section, we propose a state-feedback control scheme, based on the APC methodology. The control design is composed of three steps. First, we select the desired output performance characteristics for the closed-loop system. Next, we design a control signal that integrates the desired output specifications while incorporating an adaptive mechanism to handle amplitude input constraints. Finally, we present the dynamic control law along with an adaptive mechanism to address rate input constraints.

\textit{Step 1.}
First we employ the linear filter:
\begin{equation}\label{slide}
    s(e(t)) \coloneqq w^Te(t) = \delta\left[ \prod_{i=1}^{n-1} \left( \frac{d}{dt} + \lambda_i \right) \right]  e_1(t)
\end{equation}
\noindent where $\delta$ denotes a small positive constant that scales the surface error $s(e(t))$ and $e(t) \coloneqq [e_1(t),\dots,e_n(t)]^T \in \mathbb{R}^n$ represents the state tracking error, with $e_i(t) \coloneqq  x_i(t) - \left(\frac{d}{dt}\right)^{i-1}y_r(t)$ for $i = 1,\ldots,n$. Additionally, $w\coloneqq \delta[\omega_1,\dots,\omega_{n-1},1]^T \in \mathbb{R}^n$ consists of the coefficients of a Hurwitz polynomial $s^{n-1} + \omega_{n-1}s^{n-2} + \dots + \omega_2s + \omega_1$, with real negative roots $-\lambda_i < 0$ for $i = 1,\dots,n-1$.

According to \cite{slotine} (pp. 277), the tracking problem of \eqref{eq:system} is equivalent to ensuring that $e(t)$ remains within the invariant boundary layer $S\coloneqq \{ e \in \mathbb{R}^n: \lvert s(\Tilde{e}) \rvert \leq \Bar{s} \}$, where $\Bar{s}$ is a positive constant. If the initial condition $s(e(0)) \in S $ then all $e_i (t) , ~ i=1,\ldots,n$ remain bounded within compact sets whose size depends on $\Bar{s}$. Consequently, constraints on \eqref{slide} can be directly translated into bounds on the tracking error vector $e(t)$, making $s(e)$ an actual output performance metric of \eqref{eq:system}. For simplicity, let $s$ denote the time-varying tracking error surface $s(e(t))$. Subsequently, the desired performance characteristics are defined through the selection of the prescribed performance parameters. In particular, we set the minimum exponential convergence rate $l_{u_0}>0$ and the maximum absolute steady state error  $e_\infty>0$.

\textit{Step 2.}
In this step, we design the nominal control signal $u_d(s,\rho_u)$ and the associated adaptive law $\dot{\rho}_u$, where $\rho_u(t)>0$ is a dynamic performance function that defines the admissible envelope for the surface error $s$ and adapts to amplitude input constraints, as follows:
\begin{align}
    & u_d(s,\rho_u)   \coloneqq -k_u \mathcal{T} \left(  \frac{s}{\rho_u} \right) \label{eq:ud} \\
    & l_u(s,\rho_u)  \coloneqq l_{u_1}(\Bar{u}-|\mathrm{sat}_{\Bar{u}}(u_d(s,\rho_u))|) + l_{u_0} \label{lu} \\[0.5ex] 
    & \dot{\rho}_u    \coloneqq  \left( \frac{\Tilde{u}_d(s,\rho_u)}{s} - l_u(s,\rho_u)  \right)\rho_u  +  l_u(s,\rho_u) \rho_{u}^{\infty} \label{eq:adaptu} \\& \hspace{1.37em} \eqcolon   \alpha_{\rho_u}(s,\rho_u)  \nonumber
\end{align}
where $\Tilde{u}_d(s,\rho_u) \coloneqq \mathrm{sat}_{\Bar{u}}(u_d(s,\rho_u))- u_d(s,\rho_u)$ represents the control deficiency caused by amplitude saturation. The function $\mathcal{T}(\cdot)$ denotes the error transformation defined in Section \ref{notations} and $k_u>0$ is a control gain. Additionally, $ \rho_{u}^{\infty} = e_\infty \delta \prod_{j=1}^{n-1} \lambda_{j}$, and $l_{u_1}$ is a positive constant affecting the convergence rate \eqref{lu}. The term \eqref{lu} is introduced to dynamically regulate the convergence rate based on the input's proximity to input saturation, exploiting the available control effort to improve the response during transient phases. Note also that $\rho_u(t) \geq \rho_{u}^{\infty}>0$, since $\frac{\Tilde{u}_d(s,\rho_u)}{s} \geq 0$ for all $t \geq 0$.

\textit{Step 3.}
In this step, we unify amplitude and rate input constraints in the control design. In particular, $\rho_r(t)>0$ is a performance function associated with the control tracking error, introduced to enforce rate input constraints:
\begin{align}
   &  u_r(s,\rho_u,\rho_r,u)   \coloneqq -k_r \mathcal{T} \left(  \frac{u - \mathrm{sat}_{\Bar{u}}(u_d(s,\rho_u))}{\rho_r} \right)\label{eq:ur} \\
    & l_r(s,\rho_u,\rho_r,u)  \coloneqq l_{r_1} (\Bar{r}-|\mathrm{sat}_{\Bar{r}}(u_r(s,\rho_u,\rho_r,u))|) + l_{r_0} \label{lr}\\[0.5ex] 
    & \dot{\rho}_r  \coloneqq \left( \frac{\Tilde{u}_r(s,\rho_u,\rho_r,u)}{u - \mathrm{sat}_{\Bar{u}}(u_d(s,\rho_u))}- l_r(s,\rho_u,\rho_r,u) \right)\rho_r \label{eq:adaptr} \\[0.5ex]  & \hspace{2.4em} +  l_r(s,\rho_u,\rho_r,u)\rho_{r}^{\infty} \eqcolon \alpha_{\rho_r}(s,\rho_u,\rho_r,u)  \nonumber
\end{align}
where:
\begin{align*}
    \Tilde{u}_r(s,\rho_u,\rho_r,u) \coloneqq \mathrm{sat}_{\Bar{r}}(u_r(s,\rho_u,\rho_r,u))- u_r(s,\rho_u,\rho_r,u)
\end{align*}
represents the control deficiency due to rate saturation and $k_r,\rho_{r}^{\infty},l_{r_0},l_{r_1}$ are positive constants. Finally, the dynamic control law, which accounts for amplitude and rate input constraints, is given by:
\begin{align}\label{eq:control}
   \dot{u} = \mathrm{sat}_{\Bar{r}}(u_r(s,\rho_u,\rho_r,u)) \triangleq \alpha_{u}(s,\rho_u,\rho_r,u). 
\end{align}
Exploiting \eqref{eq:ud} and \eqref{eq:ur} the control law \eqref{eq:control} can be also written as:
\begin{align*}
    \dot{u}  = -\mathrm{sat}_{\Bar{r}}\left(k_r \mathcal{T} \left(  \frac{u + \mathrm{sat}_{\Bar{u}}\left( k_u \mathcal{T} \left(  \frac{s}{\rho_u} \right)\right)}{\rho_r} \right)\right)  
\end{align*}
Note that that the signal $u$ that will be implemented by the actuator undergoes a rate limit imposed by $\mathrm{sat}_{\Bar{r}}(\cdot)$ as well as an amplitude limit imposed by $\mathrm{sat}_{\Bar{u}}(\cdot)$. Consequently, the signals $u$ and $\dot{u}$ remain within the feasible compact sets $\mathcal{U}$ and $\mathcal{R}$, respectively.

If we had full access to the states of the system as well as the higher order derivatives of the reference trajectory $y_r(t)$, we could leverage the proposed state-feedback controller to obtain a closed-loop system with adaptive performance attributes. The following theorem establishes a sufficient condition for the boundedness of closed-loop signals under the proposed control law, ensuring robust tracking performance in the presence of input constraints.

\begin{thm}\label{lemma1}
Consider system \eqref{eq:system} under Assumptions \ref{ass1} and \ref{ass2}, and assume a reference trajectory \( y_r(t) \) that satisfies Assumption \ref{ass3}. Assume also that $u(0)\in[-\Bar u,\Bar u]$, $|s(e(0))|<\rho_u(0)$, $|u(0)-\operatorname{sat}_{\Bar u}(u_d(0))|<\rho_r(0)$ with $\rho_i(0)\geq\rho_i^\infty>0$, $i\in\{u,r\}$. Suppose that the fixed actuator limits $(\Bar u,\Bar r)$ satisfy
the feasibility condition \eqref{eq:actuator-feasibility},
specified in Phase B of the proof. Then the state-feedback
controller \eqref{eq:control}, together with the adaptive laws
\eqref{eq:adaptu} and \eqref{eq:adaptr}, ensures that:
\begin{enumerate}
\item All closed-loop signals remain bounded for all $t\geq0$,
with $u(t)\in\mathcal U$ and $\dot u(t)\in\mathcal R$.
\item There exist finite constants $\Bar{\rho}_u>0$ and
$\Bar{\xi}_u\in(0,1)$ such that
\begin{equation}\label{Sset}
|s(e(t))|
\leq\Bar{\xi}_u\rho_u(t)
<\rho_u(t)<\Bar{\rho}_u,
\qquad t\geq0.
\end{equation}
\end{enumerate}
\end{thm}
\begin{proof}
  Let us first define the normalized tracking error $\xi_u(t)  \coloneqq \frac{s(e(t))}{\rho_u(t)}$. Thenceforward, we omit the arguments of the closed-loop signals unless it is necessary, for clarity in the proof. Differentiating $\xi_u$ with respect to time and invoking \eqref{slide} we get: 
  \begin{align}\label{eq:dxi0}
      \dot{\xi}_u=\frac{1}{\rho_u} \left( w^T[\dot{x} - \dot{x}_d] - \xi_u \dot{\rho}_u \right).
  \end{align}
By denoting $x = e + x_d$ and substituting \eqref{eq:system} and \eqref{eq:adaptu} into \eqref{eq:dxi0} we obtain:
\begin{align}
    \dot{\xi}_u=\frac{1}{\rho_u} \Bigg(& \delta [f( e + x_d, z, d)  + g( e + x_d,z,d)u  - \left( \frac{d}{dt} \right)^n y_r \nonumber \\ & + \sum \limits _{j=1}^{n-1}\omega_je_{j+1}] + \xi_u l_u \left( \rho_u - \rho_{u}^{\infty}  \right) -  \Tilde{u}_d \Bigg) . \label{eq:dxia}
\end{align}
where the control deficiency is defined as $\Tilde{u}_d=\mathrm{sat}_{\Bar{u}}(u_d)- u_d$. Moreover, let us define the error $e_u \coloneqq u - \mathrm{sat}_{\Bar{u}}(u_d)$. Next, let us substitute \eqref{lu} as well as add and subtract $\delta g( e + x_d,z,d)\mathrm{sat}_{\Bar{u}}(u_d)$ into \eqref{eq:dxia}, leading to:
\begin{equation}\label{eq:dxi}
\begin{split}
        \dot{\xi}_u=\frac{1}{\rho_u} \left(h +  u_d \right)
    \end{split}
\end{equation}
where:
\begin{align*}
    h\coloneqq &\delta [f( e + x_d, z, d)  + g( e + x_d,z,d)e_u  -\left( \frac{d}{dt} \right) ^n y_r \\ & + \sum \limits _{j=1}^{n-1}\omega_je_{j+1}] + \xi_u (l_{u_1}\Delta\Bar{u}+ l_{u_0}) \left( \rho_u - \rho_{u}^{\infty} \right) \\ &+( \delta g( e + x_d,z,d)- 1)\mathrm{sat}_{\Bar{u}}(u_d)
\end{align*}
with $\Delta\Bar{u}\coloneqq\Bar{u}-|\mathrm{sat}_{\Bar{u}}(u_d)| $.   Following the same reasoning, we define the normalized error $\xi_r \coloneqq \frac{e_u}{\rho_r}$ and substituting \eqref{eq:ud},\eqref{lr} and \eqref{eq:adaptr}, the dynamics of $\xi_r$ is given by:
\begin{equation}\label{dxir}
    \begin{split}
        \dot{\xi}_r=\tfrac{1}{\rho_r}\left( S(\xi_u) \dot{\xi}_u      +\xi_r (l_{r_1}\Delta\Bar{r}+ l_{r_0}) \left( \rho_r - \rho_{r}^{\infty} \right)  +   u_r \right)
    \end{split}
\end{equation}
with $S(\xi_u) \coloneqq k_u\mathrm{sat}^{'}_{\Bar{u}}\left(-k_u \mathcal{T}(\xi_u)\right) \mathcal{T}^{'}(\xi_u)$ and $\Delta\Bar{r}\coloneqq\Bar{r}-|\mathrm{sat}_{\Bar{r}}(u_r)| $.
To continue, consider candidate constants
$\Bar{\rho}_u>\rho_u(0)$ and
$\Bar{\rho}_r>\rho_r(0)$, and define $\mathcal C\coloneqq
\mathbb R^{m+n}\times(-1,1)\times(0,\Bar{\rho}_u)
\times(-1,1)\times(0,\Bar{\rho}_r) .$ The augmented state
$\zeta=[z^T,e^T,\xi_u,\rho_u,\xi_r,\rho_r]^T$ satisfies
$\zeta(0)\in\mathcal C$. The proof proceeds in three phases.

\textit{Phase A. }
The closed-loop dynamics can be written as
$\dot\zeta=\phi(t,\zeta)$, where
$\phi:\mathbb R_+\times\mathcal C\to\mathbb R^{m+n+4}$ is
piecewise continuous and locally integrable in time and locally
Lipschitz in $\zeta$. Hence, Theorem 54 in \cite{ds} guarantees
a unique maximal solution
$\zeta:[0,\tau_{\max})\to\mathcal C$ for all \(t \in [0,\tau_{\max})\) with \(\tau_{max} \in \{ \mathbb{R}_{+}^*, \infty \}\).

\textit{Phase B.} Since $\zeta(t)\in\mathcal C$, the transformed errors
$\epsilon_i\coloneqq T(\xi_i)$, $i\in\{u,r\}$, are well defined
on $[0,\tau_{\max})$. Moreover, $|s|<\Bar{\rho}_u$. The
stability of the linear filter \eqref{slide}, together with the initialization set $\mathcal X_0$ and the reference set $\mathcal Y$, yields a compact
set $\Omega_e$ containing $e$. Consequently,
$x=e+x_d\in\Omega_e\oplus\mathcal Y$, and
Assumption~\ref{ass2}, together with
$z(0)\in\mathcal Z_0$, yields a compact set $\Omega_z$
containing $z$. Furthermore, $\mathcal U$ is positively invariant. Indeed, at
$u=\Bar u$, one has $e_u\geq0$, and hence $\dot u\leq0$; at
$u=-\Bar u$, one has $e_u\leq0$, and hence $\dot u\geq0$.
Therefore, $u(t)\in\mathcal U$ and $|e_u(t)|\leq2\Bar u$.

The preceding bounds and the local Lipschitz continuity of $f$
and $g$ imply that $|h|\leq H_u$ for some finite $H_u>0$. Thus,
for $V_u=\frac12\epsilon_u^2$, \eqref{eq:dxi} gives:
\begin{align}\label{dotvu}
    \dot V_u\leq
\frac{1}{\rho_u}\left[
H_u|\mathcal T(\xi_u)|
-k_u|\mathcal T(\xi_u)|^2
\right].
\end{align}
Choose a constant $\Bar\epsilon_u$ such that $\Bar\epsilon_u>
\max\left\{
|\epsilon_u(0)|,\frac{H_u}{k_u}
\right\}$.
Define
\[
m_u^s\coloneqq
\frac{
\Bar\epsilon_u
(k_u\Bar\epsilon_u-H_u)}
{\Bar\rho_u}>0.
\]
At $|\epsilon_u|=\Bar\epsilon_u$, inequality
\eqref{dotvu} gives $\dot V_u\leq-m_u^s<0$, which leads to $|\xi_u(t)|<\Bar\xi_u
\coloneqq T^{-1}(\Bar\epsilon_u)<1,~\forall t\in[0,\tau_{\max})$.
Define:
\[
M_u(\Bar u)\coloneqq
\max_{|\xi|\leq\Bar\xi_u}
\frac{
\operatorname{sat}_{\Bar u}(-k_u\mathcal T(\xi))
+k_u\mathcal T(\xi)}{\xi}
\]
where the quotient is continuously extended by zero at
$\xi=0$. It follows from \eqref{eq:adaptu} that $\dot\rho_u\leq
M_u(\Bar u)-l_{u_0}(\rho_u-\rho_u^\infty).$

Similarly, the boundedness of
$\xi_u,\rho_u,e,z,u$ implies that
$S(\xi_u)\dot\xi_u$ is bounded. Hence, for some finite
$H_r>0$,
\[
\left|
S(\xi_u)\dot\xi_u+
\xi_r(l_{r_1}\Delta\Bar r+l_{r_0})
(\rho_r-\rho_r^\infty)
\right|\leq H_r.
\]
Therefore, using the Lyapunov function $V_r = \tfrac{1}{2}\epsilon_r^2$ we get $|\epsilon_r(t)|\leq\Bar\epsilon_r
\coloneqq
\max\left\{|\epsilon_r(0)|,\frac{H_r}{k_r}\right\}$ and $|\xi_r(t)|\leq\Bar\xi_r
\coloneqq T^{-1}(\Bar\epsilon_r)<1.$
Next define:
\[
M_r(\Bar r)\coloneqq
\max_{|\xi|\leq\Bar\xi_r}
\frac{
\operatorname{sat}_{\Bar r}(-k_r\mathcal T(\xi))
+k_r\mathcal T(\xi)}{\xi}
\]
again using the continuous extension at $\xi=0$. Then $\dot\rho_r\leq
M_r(\Bar r)-l_{r_0}(\rho_r-\rho_r^\infty).$

The fixed actuator pair $(\Bar u,\Bar r)$ is called feasible if
there exist auxiliary constants
$\Bar\rho_u>\rho_u(0)$ and
$\Bar\rho_r>\rho_r(0)$ satisfying
\begin{equation}\label{eq:actuator-feasibility}
M_u(\Bar u)
<l_{u_0}(\Bar\rho_u-\rho_u^\infty),
~~
M_r(\Bar r)
<l_{r_0}(\Bar\rho_r-\rho_r^\infty).
\end{equation}
Under \eqref{eq:actuator-feasibility}, there exist
$\delta_u,\delta_r>0$ such that $\rho_i^\infty\leq\rho_i(t)
\leq\Bar\rho_i-\delta_i$ with $i\in\{u,r\},\forall t\in[0,\tau_{\max}).$

\textit{Phase C. }The preceding estimates imply that the maximal solution remains
in the compact set $\mathcal C'\coloneqq
\Omega_z\times\Omega_e
\times[-\Bar{\xi}_u,\Bar{\xi}_u]
\times[\rho_u^\infty,\Bar{\rho}_u-\delta_u]
\times[-\Bar{\xi}_r,\Bar{\xi}_r]
\times[\rho_r^\infty,\Bar{\rho}_r-\delta_r]$
which is strictly contained in $\mathcal C$. Proposition C.3.6
in \cite{ds} therefore yields $\tau_{\max}=\infty$.
Consequently, all closed-loop signals remain bounded,
$u(t)\in\mathcal U$, $\dot u(t)\in\mathcal R$, and
\eqref{Sset} follows from $|s|=|\xi_u|\rho_u$, completing the
proof.
\end{proof}

\begin{rem}\label{satlevel}
Condition~\eqref{eq:actuator-feasibility} is a joint feasibility
condition on the fixed actuator limits $(\Bar u,\Bar r)$. It
requires sufficient amplitude and rate authority for finite
auxiliary bounds $\Bar\rho_u$ and $\Bar\rho_r$ to exist, thereby
rendering the corresponding closed-loop compact set positively
invariant. These bounds depend on the compact
initialization sets, the reference and disturbance bounds, the
system nonlinearities over the induced compact set, and the
coupling between the amplitude and rate tracking loops.
An a priori computation of $\Bar\rho_u$ and $\Bar\rho_r$ in presence of hard input constraints remains an open problem and is left for future work.
\end{rem}
\begin{rem}\label{APF}
     In contrast to the related literature on APC \cite{smc,ratesiso,bergerIC} in this work we propose a novel, input-dependent convergence rate that adjusts automatically based on the available control effort. The term \eqref{lu} regulates the convergence rate according to the system's distance from the input saturation limit, ensuring that the system retains the prescribed steady-state performance specifications as fast as possible, leveraging the maximum feasible control effort during the transient period. The adaptive performance funnel (APF) refers to the time-varying constraint set $\mathcal{S}$ defined in \eqref{Sset}, within which the surface error $s(e(t))$, is constrained to evolve. Figure~\ref{fb} compares the APF generated by the proposed scheme, denoted by $\mathcal{F}_P(t)$ with the corresponding tracking error $e_P(t)$, against the APF of our previous work in \cite{ratesiso}, denoted by $\mathcal{F}_T(t)$ with tracking error $e_T(t)$. Notice that the use of the input-dependent rate \eqref{lu} leads to a less conservative performance funnel and faster transient recovery by explicitly leveraging the available control authority.
     \begin{figure}[thpb]
      \centering  \includegraphics[width=\columnwidth]{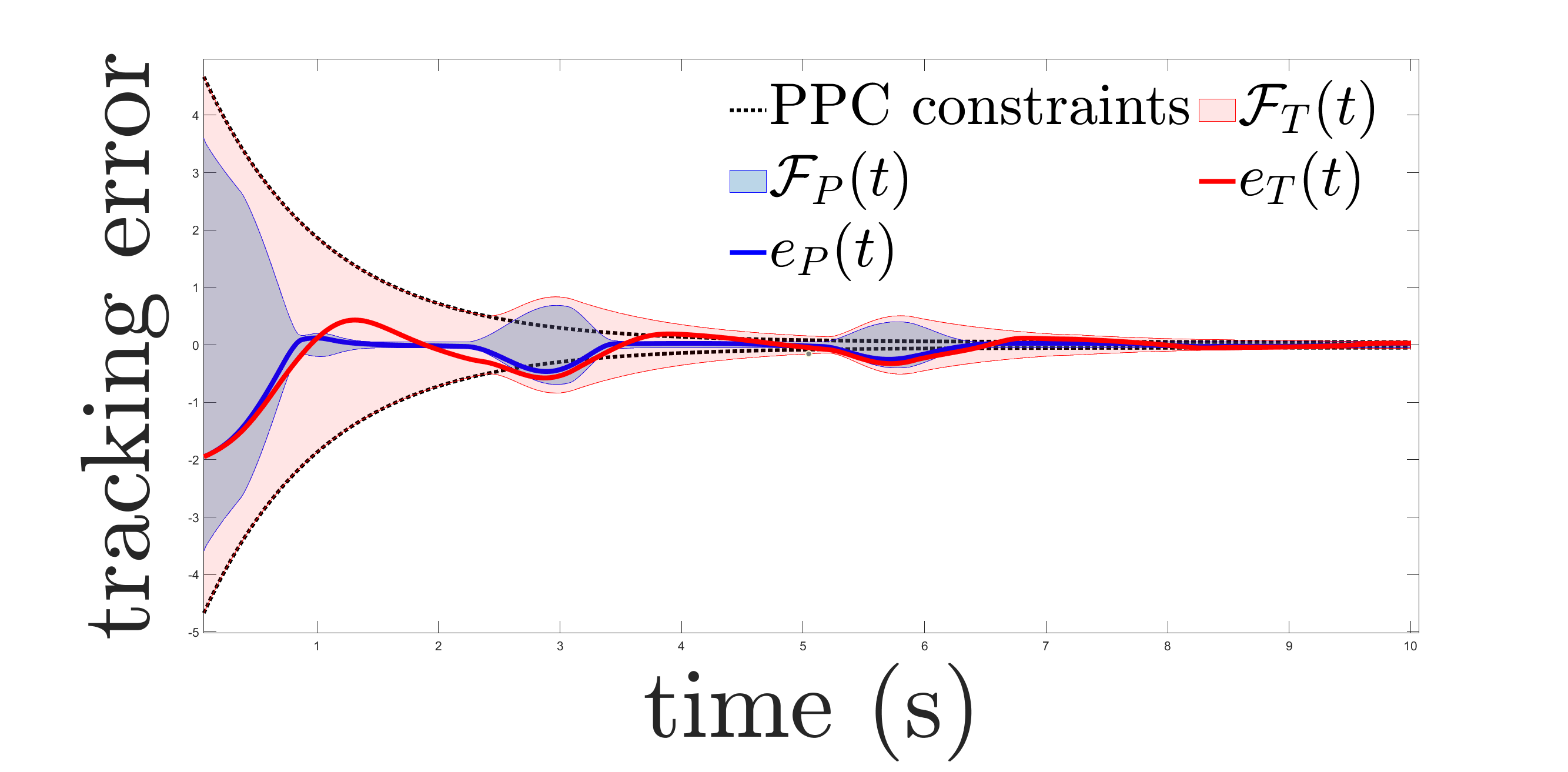}
   \caption{Comparison of adaptive performance funnels: The proposed scheme $\mathcal{F}_P(t)$ vs. the approach $\mathcal{F}_T(t)$ from \cite{ratesiso} for the benchmark system simulated in \cite{ratesiso}.}
      \label{fb}
\end{figure}
\end{rem}
The following Corollary establishes sufficient conditions to ensure accurate tracking, while preventing input saturation and guaranteeing the exponential decay of the performance functions to their nominal steady-state values. Specifically, it extends Theorem \ref{lemma1} by providing sufficient conditions on the input saturation levels \( \Bar{u} \) and \( \Bar{r} \) to guarantee that the control system can deliver the necessary control effort for tracking while adhering to the predefined performance constraints.
\begin{cor}\label{cor1}
Suppose that Theorem~\ref{lemma1} holds and let $\tau_0\geq0$
satisfy $|u_d(\tau_0)|<\Bar u-\beta$ and
$|u_r(\tau_0)|<\Bar r-\beta$, where $\beta>0$ is the smoothing
width of $\operatorname{sat}_{\sigma}$. For $t\geq\tau_0$, define:
\begin{align*}
q_u\coloneqq{}&
f-y_r^{(n)}+\sum_{j=1}^{n-1}\omega_je_{j+1}+ge_u
+\frac{\xi_ul_u}{\delta}(\rho_u-\rho_u^\infty)\\
q_r\coloneqq{}&
S(\xi_u)\dot\xi_u+\xi_rl_r(\rho_r-\rho_r^\infty)
\end{align*}
where $e_u=u-\operatorname{sat}_{\Bar u}(u_d),$
$\xi_u=s/\rho_u$, and $\xi_r=e_u/\rho_r$, and define:
\begin{align}
\Bar u_a&\coloneqq
\beta+\sup_{t\geq\tau_0}
\left|\frac{q_u(t)}{g(t)}\right|
\label{uac}\\
\Bar r_a&\coloneqq
\beta+\sup_{t\geq\tau_0}|q_r(t)|.
\label{uac1}
\end{align}
If $\Bar u>\Bar u_a$ and $\Bar r>\Bar r_a$, then $u_d$
and $u_r$ remain unsaturated for all $t\geq\tau_0$. Moreover:
\[
\operatorname{dist}\bigl(s(e(t)),\mathcal S_\infty\bigr)
\leq
\bigl(\rho_u(\tau_0)-\rho_u^\infty\bigr)
e^{-l_{u_0}(t-\tau_0)}
\]
with $\mathcal S_\infty
\coloneqq
\{s\in\mathbb R:|s|\leq\rho_u^\infty\}.$
\end{cor}
\begin{proof}
As long as both saturation functions remain inactive,
$\operatorname{sat}_{\Bar u}(u_d)=u_d$ and
$\operatorname{sat}_{\Bar r}(u_r)=u_r$. Consequently:
\[
\dot\xi_u
=
\frac{\delta g}{\rho_u}
\left(u_d+\frac{q_u}{g}\right),
\qquad
\dot\xi_r
=
\frac{1}{\rho_r}(u_r+q_r).
\]
At the first possible activation boundary
$|u_d|=\Bar u-\beta$, condition~\eqref{uac} implies
$\xi_u\dot\xi_u<0$, because
$u_d=-k_u\mathcal T(\xi_u)$ has sign opposite to $\xi_u$.
Similarly, at $|u_r|=\Bar r-\beta$.,
condition~\eqref{uac1} implies $\xi_r\dot\xi_r<0$.
Thus, neither signal can leave its unsaturated region.

The saturation deficiencies therefore vanish, and the adaptive
laws reduce to $\dot\rho_i=-l_i(t)(\rho_i-\rho_i^\infty),~i\in\{u,r\}.$ In particular:
\[
\rho_u(t)-\rho_u^\infty
=
\bigl(\rho_u(\tau_0)-\rho_u^\infty\bigr)
\exp\left(-\int_{\tau_0}^{t}l_u(\tau)\,d\tau\right).
\]
Since $l_{u_0}\leq l_u(t)\leq l_{u_1}\Bar u+l_{u_0}$ the exponential bounds are:
\[
e^{-(l_{u_1}\Bar u+l_{u_0})(t-\tau_0)}
\leq
\exp\left(-\int_{\tau_0}^{t}l_u(\tau)\,d\tau\right)
\leq
e^{-l_{u_0}(t-\tau_0)}.
\]
Finally, $|s(e(t))|<\rho_u(t)$ by
Theorem~\ref{lemma1}, which gives the stated distance bound and
completes the proof.
\end{proof}
\begin{lem}\label{propPPC}
Suppose that Corollary \ref{cor1} holds for all $t \geq \tau_0$ and consider the metric $s(e(t)) \in \mathbb{R}$ as defined in \eqref{slide}. If the parameters of $s(e(t))$ are chosen such that $l_{u_0} < \min \limits _{j = 1,\dots,n-1} \{  \lambda_{i}  \}$, then for each component $e_i(t)$ of the error $e(t)$ there exist positive constants $\Bar{e}_i $ such that: 
    \begin{equation*}
        \left | e_i (t)\right | \leq \Bar{e}_i \exp (-l_{u_0}(t-\tau_0)) +  \frac{\rho_u^{\infty}2^{i-1}}{\delta \prod  \limits _{j=1}^{n-i}\lambda_{j}}\text{, } \forall t \geq \tau_0. 
    \end{equation*} 
\end{lem}
\begin{proof}
The proof follows the same steps as Lemma 1 in \cite{dimanidis}. 
\end{proof}
Lemma \ref{propPPC} states that with an appropriate choice of the filter \( s(\cdot) \), the minimum convergence rate of the output tracking error \( e_1(t) \) can be explicitly determined, in the absence of input saturation.
\begin{figure*}[thpb]
\centering
\input{BD}
\caption{Block diagram of the proposed controller \eqref{PPOE}-\eqref{eq:outcontrol}.}
\label{fa}
\end{figure*}
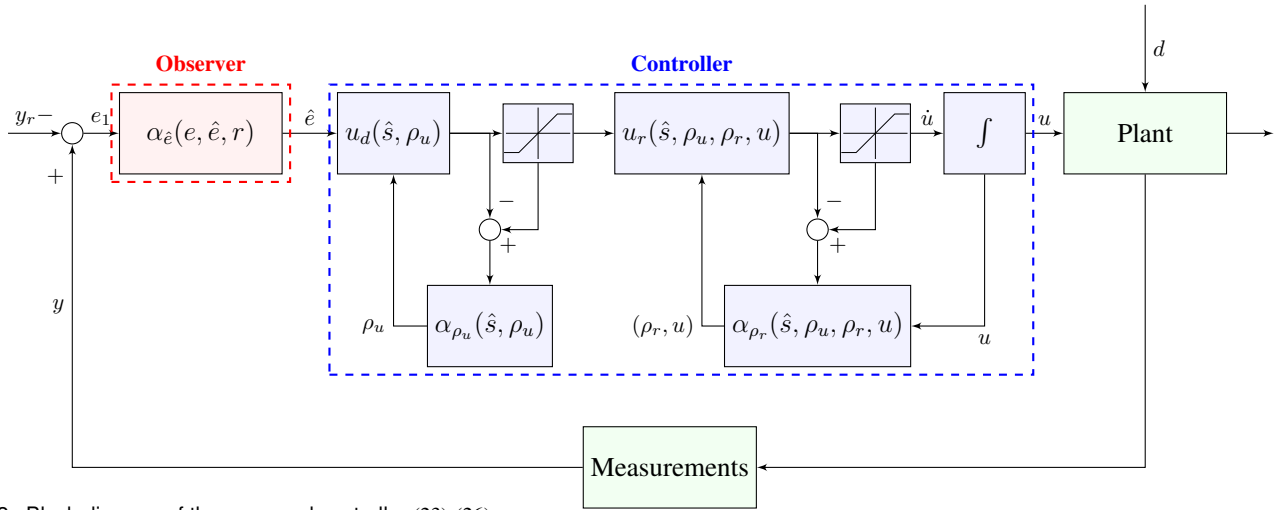

\subsection{Prescribed Performance Observer}\label{PPO}
 In this section, we introduce the PPO, an advancement of the HGO \cite{hgo,hgo8} with nonlinear gains that facilitate its design and enhance its performance. Leveraging the measurable output tracking error $e_1(t)$ and building on the PPC technique~\cite{bechltac}, we propose the following observer:
\begin{align}
       &\dot{\hat{e}}_i \coloneqq \hat{e}_{i+1} + \frac{\gamma_i r(t)}{r_\infty^i}  T \left( \frac{e_1(t)-\hat{e}_1(t)}{r(t)} \right)~,i=1,\dots,n-1 \nonumber \\
    &\dot{\hat{e}}_n \coloneqq  \frac{\gamma_n r(t)}{r^n_\infty}T\left( \frac{e_1(t)-\hat{e}_1(t)}{r(t)} \right) \label{PPOt}
\end{align}
where $T(\cdot)$ denotes the error transformation defined in Section \ref{notations}, $\gamma_i$ are positive constants and $r(t)$ is a user-specified performance function:
\begin{align*}
    r(t) \coloneqq (r(0)- r_{\infty})\mathrm{exp}(-l_o t)+r_{\infty} ,~l_o>0 , r_{\infty} \in (0,1)
\end{align*}
with $r(0) > |e_1(0)-\hat{e}_1(0)|$ and $r(0) \geq r_{\infty} $. By denoting $\hat{e} \coloneqq [\hat{e}_1,\dots,\hat{e}_n]^T$ the proposed observer can be written compactly as: 
\begin{align}\label{eq:PPO}
     \dot{\hat{e}}  \coloneqq \alpha_{\hat{e}}(e,\hat{e},r)
\end{align}
The PPO is nonlinear but features a simple structure that does not rely on system knowledge, approximation or adaptive mechanisms, maintaining the computational simplicity of a standard HGO. 

\subsubsection{Parameters Selection} The key characteristic of \eqref{eq:PPO} concerns the nonlinear gains that facilitate the parameter selection and enhance the performance of the observer. First, $\gamma_i,~i=1,\dots, n$ are selected such that the polynomial:
\begin{equation}\label{obspar}
    p(\mathfrak{s}) = \mathfrak{s}^n + \gamma_1\mathfrak{s}^{n-1} + \dots + \gamma_{n-1}\mathfrak{s} + \gamma_n
\end{equation}
is Hurwitz. The remaining observer parameters determine the prescribed-performance
function $r(t).$ Let $E_o\coloneqq
\sup_{\mathcal X_0\times\mathcal Y}
|e_1(0)-\hat e_1(0)|.$ For a certain initialization margin $\mu_o\in(0,1)$, selecting
$r(0)>E_o/\mu_o$ guarantees
$|e_1(0)-\hat e_1(0)|/r(0)<\mu_o$. Thus, $r(0)$ should not be
chosen excessively close to the initial estimation error, since this
would initialize the normalized error near the singular boundary of
$\mathcal T$.

The parameter $r_\infty$ specifies the ultimate output-estimation
accuracy, since
$\limsup_{t\to\infty}|e_1(t)-\hat e_1(t)|\leq r_\infty$.
A smaller $r_\infty$ improves this bound but increases the observer
bandwidth, numerical stiffness, and sensitivity to measurement noise.
Finally, $l_o$ determines how rapidly the performance boundary
approaches $r_\infty$. In particular, if
$r_\infty<r_c<r(0)$ is required at time $t_c>0$, it is sufficient to
select:
\[
l_o\geq\frac{1}{t_c}
\ln\!\left(\frac{r(0)-r_\infty}{r_c-r_\infty}\right).
\]
Larger values of $l_o$ provide faster contraction but also increase
the initial boundary rate
$|\dot r(0)|=l_o\bigl(r(0)-r_\infty\bigr)$ and should therefore be
selected consistently with the stability condition of
Theorem~\ref{th2}.

\begin{rem}
The parameter $r(0)$ must accommodate the initial estimation error, while $r_\infty$ sets its desired steady-state bound. Increasing $l_o$ accelerates convergence of $r(t)$, whereas excessively small $r_\infty$ or large $l_o$ may increase noise sensitivity and numerical stiffness.
\end{rem}

\subsubsection{Computational Complexity}

We compare the number of observer states and scalar operations required
per vector-field evaluation. The constant gain coefficients are
precomputed, and the prescribed function $r(t)$ is calculated offline.
Operations common to all observers are omitted. In
Table~\ref{tab:observer_complexity}, $\mathrm{M}$,
$\mathrm{A}$, $\mathrm{D}$, and $\mathrm{T}$ denote multiplication,
addition/subtraction, division, and one evaluation of the scalar
mapping $T$, respectively.
\begin{table}[t]
\centering
\caption{Complexity per observer vector-field evaluation}
\label{tab:observer_complexity}
\setlength{\tabcolsep}{3pt}
\renewcommand{\arraystretch}{1.1}
\begin{tabular}{lccc}
\toprule
\textbf{Design} & \textbf{States} & \textbf{Scalar operations} & \textbf{Order} \\
\midrule
HGO
& $n$
& $n\mathrm{M}+n\mathrm{A}$
& $O(n)$ \\
PPO
& $n$
& $n\mathrm{M}+n\mathrm{A}+\mathrm{D}+\mathrm{T}$
& $O(n)$ \\
Standard ESO
& $n+1$
& $(n+1)\mathrm{M}+(n+1)\mathrm{A}$
& $O(n)$ \\
\bottomrule
\end{tabular}
\end{table}
Thus, the PPO has the same $O(n)$ arithmetic complexity and the same
number of differential equations as a classical HGO. Its additional online cost relative to an HGO consists of one scalar
division and one evaluation of $T$. Compared with a standard ESO that introduces an additional state to estimate the lumped uncertainty, the PPO uses one fewer observer state. Therefore, the PPO has computational complexity comparable to that of an HGO,
while providing a modest dimensional advantage over a standard ESO.

\subsubsection{Gain Structure and Steady-State Noise Amplification} \label{gainAmp}
For the $i$-th PPO injection term:
\[
\phi_i \coloneqq
\frac{\gamma_i r(t)}{r_\infty^i}
 T\!\left(
\frac{e_1-\hat e_1}{r(t)}
\right)
\]
the differential gain with respect to the output-estimation error is:
\[
\frac{\partial\phi_i}{\partial(e_1-\hat e_1)}
=
\frac{\gamma_i}{r_\infty^i}J(\xi_o),
\quad
\xi_o\coloneqq\frac{e_1-\hat e_1}{r(t)}
\]
where $J(\xi_o)= T'(\xi_o)$. Since
$J(\xi_o)\to1$ as $\xi_o\to0$, the highest local steady-state gain is
$O(r_\infty^{-n})$, as in the classical HGO. Thus, although the PPO provides a continuously varying gain structure,
it does not reduce the local steady-state noise-amplification order.
The smoother control response observed in the simulations is
atributed to the explicit rate-limiting mechanism, which directly
attenuates rapid variations in the applied input.

\section{Main Results}\label{main}
In this section, we first present an output-feedback APC scheme and subsequently we prove performance recovery of the state-feedback controller  \eqref{eq:adaptu},\eqref{eq:adaptr},\eqref{eq:control} when interconnected with the proposed PPO \eqref{eq:PPO}, presenting a separation principle.
\subsection{Output-Feedback Controller}
In the output feedback implementation, all state-feedback terms in \eqref{eq:adaptu},\eqref{eq:adaptr},\eqref{eq:control}, involving \( e \) are replaced by their corresponding output-feedback counterparts. Specifically, the term \( s \) is replaced by its saturated approximation:
\begin{align*}
    \hat{s}(t)\coloneqq
\rho_u(t)\operatorname{sat}_{1-\xi_s}
\left(\frac{s(\hat e(t))}{\rho_u(t)}\right),
\qquad 0<\xi_s<1-\beta.
\end{align*}
The resulting output-feedback controller is expressed as follows:

\begin{align}
    & \dot{\hat{e}}  = \alpha_{\hat{e}}(e,\hat{e},r) \label{PPOE}\\
           & \dot{\rho}_u  = \alpha_{\rho_u}(\hat{s},\rho_u)  \label{eq:outadaptu}  \\
           & \dot{\rho}_r  = \alpha_{\rho_r}(\hat{s},\rho_u,\rho_r,u)  \label{eq:outadaptr}  \\
           &  \dot{u} = \alpha_{u}(\hat{s},\rho_u,\rho_r,u) \label{eq:outcontrol} 
\end{align} 
Saturating observer-based signals outside a compact set is a
standard approach for mitigating the peaking phenomenon of
HGOs~\cite{hgo8,dimanidis,funnel}. In this paper, the normalized
saturation guarantees $\left|\frac{\hat s(t)}{\rho_u(t)}\right|
\leq1-\xi_s<1$, so the transformation in the output-feedback controller remains
well defined throughout the observer transient. The block diagram of the proposed controller is
presented in Fig.~\ref{fa}, while Table~\ref{tab:tuning}
summarizes the main design parameters and tuning guidelines.
\begin{table*}[t]
\centering
\caption{Controller design parameters and practical tuning guidelines}
\label{tab:tuning}
\setlength{\tabcolsep}{6pt}
\renewcommand{\arraystretch}{1.1}
\begin{tabular}{p{2.0cm} p{6.4cm} p{8.03cm}}
\toprule
\textbf{Parameter}
& \textbf{Role in the control design}
& \textbf{Practical tuning guideline} \\
\midrule
$\lambda_i,\ \delta$
& Shape the surface-error dynamics and the scaling of $s(e(t))$
& Select according to \eqref{slide} to achieve the desired nominal transient response. \\

$\rho_u(0),\ \rho_r(0)$
& Initial performance envelopes
& Choose such that
$\rho_u(0)>|s(e(0))|$ and
$\rho_r(0)>
|u(0)-\operatorname{sat}_{\bar u}(u_d(0))|$. \\

$\rho_u^\infty,\ \rho_r^\infty$
& Prescribed steady-state performance bounds
& Set $\rho_u^\infty$ according to the desired steady-state error,
$\rho_u^\infty=e_\infty\delta\prod_{j=1}^{n-1}\lambda_j$,
and select $\rho_r^\infty$ sufficiently small to achieve the desired
control-tracking accuracy. \\

$k_u,\ k_r$
& Nominal and rate-control gains
& Higher values accelerate convergence but may trigger saturation earlier. \\

$l_{u_0},\ l_{r_0}$
& Minimum convergence rates of $\rho_u(t)$ and $\rho_r(t)$
& Set the prescribed minimum decay rates of the adaptive performance
functions. \\

$l_{u_1},\ l_{r_1}$
& Convergence acceleration away from amplitude and rate saturation
& Larger values exploit the available control authority more aggressively
but may trigger saturation earlier. \\

$\gamma_i$
& PPO observer gains
& Select such that \eqref{obspar} is Hurwitz and
the sector condition \eqref{obssector} is feasible over the
selected observer-error sector.\\

$r(0)$
& Initial observer performance boundary
& Choose $r(0)>E_o/\mu_o$, where $\mu_o\in(0,1)$ is the desired initial
normalized-error margin. \\

$r_\infty$
& Ultimate output-estimation bound
& Select according to the desired accuracy and Theorem~\ref{th2};
smaller values increase noise sensitivity and numerical stiffness. \\

$l_o$
& Contraction rate of $r(t)$
& For a target $r(t_c)\leq r_c$, select
$l_o\geq t_c^{-1}
\ln\!\big((r(0)-r_\infty)/(r_c-r_\infty)\big)$,
subject to Theorem~\ref{th2}. \\

$\xi_s$
& Normalized saturation margin for $s(\hat e)$
& Select such that
$0<\xi_s<1-\beta-\bar{\xi}_u$; see Remark~\ref{bars}. \\
\bottomrule
\end{tabular}
\end{table*}

\begin{rem}\label{bars}
The saturation $\operatorname{sat}_{1-\xi_s}$ is the identity
for $|\chi|<1-\xi_s-\beta$. Since Theorem~\ref{lemma1}
guarantees $|s(e)|/\rho_u\leq\Bar{\xi}_u$, selecting
$\Bar{\xi}_u<1-\xi_s-\beta$ ensures that the additional
saturation is inactive on the state-feedback invariant set.
During the observer transient, however, it guarantees
$|\hat{s}|/\rho_u\leq1-\xi_s<1$, so that the transformation
$\mathcal T(\hat{s}/\rho_u)$ remains well defined.
\end{rem}

\begin{rem}\label{MIMO}  
 The proposed control strategy can be extended to feedback linearizable MIMO systems that, under some additional technical assumptions, can be transformed into the Byrnes-Isidori normal form \cite{isidorib}. Specifically, consider an $n$-th order MIMO nonlinear system with \(m\) inputs and \(m\) outputs and uniform relative degree \(r_u = \{r_1, \dots, r_m\}\), described compactly as:  
\begin{equation} \label{eq:systemMIMO}
    \begin{split}
    & \mathbf{\dot{z}} = \mathbf{f_0(x,z,d)} \\
    &\mathbf{\dot{x} = Ax + B[f(x,z,d) + G(x,z,d)u]} \\
    & \mathbf{y = C^Tx}
    \end{split}
\end{equation} 
with the block diagonal matrices given by:
\begin{align*}
    \mathbf{A} &= \mathrm{blockdiag}\left( A_1 ,\dots , A_m  \right) ,~ A_i = \begin{bmatrix}
    \textbf{0}_{r_i-1} & \textbf{\textit{I}}_{r_i-1} \\ 0 & \textbf{0}_{r_i-1}^T
\end{bmatrix} \\
    \mathbf{B} &=\mathrm{blockdiag}\left( B_1 ,\dots , B_m  \right),~ B_i =\begin{bmatrix}
    \textbf{0}_{r_i-1} \\ 1
\end{bmatrix} \\
    \mathbf{C} &= \mathrm{blockdiag}\left(C_1 ,\dots , C_m  \right) ,~ C_i =\begin{bmatrix}
    1 \\ \textbf{0}_{r_i-1}
\end{bmatrix}
\end{align*}
where \( \mathbf{x} = [x_1^T, \dots, x_m^T]^T \in \mathbb{R}^{r_u}\) is the state vector with \(x_i = [x_i^1, \dots, x_i^{r_i}]^T \in \mathbb{R}^{r_i}\); \(\mathbf{u} = [u_1, \dots, u_m]^T \in \overbrace{\mathcal{U} \times \dots \times \mathcal{U}}^{\text{m times}} \) is the control input; \( \mathbf{y} = [x_1^1, \dots, x_m^1]^T \in \mathbb{R}^m\) is the system output; \( \mathbf{d} \in  \mathbb{R}^p\) models bounded disturbances; and \(\mathbf{z} \in \mathbb{R}^{n-r_u}\) represents the BIBS internal dynamics. The system nonlinearities \(\mathbf{f_0(x,z, d)}\in \mathbb{R}^{n-r_u},~\mathbf{f(x,z, d)}\in \mathbb{R}^{m}\) and \(\mathbf{G( x,z, d)} \in \mathbb{R}^{m \times m}\) are locally Lipschitz, with \(\mathbf{G_s}(\chi) = \frac{\mathbf{G}(\chi) + \mathbf{G}^T(\chi)}{2}\) (the symmetric part of the input gain matrix) assumed to be positive (or negative) definite. By defining the tracking error surface \eqref{slide} elementwise for a desired trajectory \( \mathbf{y_r} = [y_{r_1}, \dots, y_{r_m}]^T\), and following the previously outlined design steps, an analogous control scheme for the MIMO case can be obtained.
\end{rem}
\subsection{Performance Recovery}
The performance recovery of the state-feedback closed-loop system, under output-feedback, is outlined in the following theorem.

\begin{thm}\label{th2}
Consider system \eqref{eq:system} under
Assumptions~\ref{ass1}--\ref{ass3} with the error metric
\eqref{slide}. Let $(\Bar u,\Bar r)$ and
$(\Bar\rho_u,\Bar\rho_r)$ satisfy
\eqref{eq:actuator-feasibility} and
\eqref{eq:of-feasibility}, respectively. Assume $\Bar\xi_u<1-\xi_s-\beta$, and that the PPO
conditions \eqref{obssector} and \eqref{obsinit} hold.
Furthermore, assume that $\hat e(0)$ belongs to a known
compact set, $u(0)\in[-\Bar u,\Bar u]$,
$\rho_i(0)\geq\rho_i^\infty>0$ for $i\in\{u,r\}$,
$|s(e(0))|<\Bar\xi_u\rho_u(0)$, and
$\left|u(0)-\operatorname{sat}_{\Bar u}(u_d^o(0))\right|
<\rho_r(0)$, where
$u_d^o(0)=u_d(\hat s(0),\rho_u(0))$.

Then, under \eqref{PPOE}--\eqref{eq:outcontrol}, there exists
$r^*>0$ such that, for every $r_\infty\in(0,r^*)$:
\begin{enumerate}
\item all closed-loop signals remain bounded;
\item $|s(e(t))|\leq\Bar\xi_u\rho_u(t)<\rho_u(t),
\qquad t\geq0;$
\item if the conditions of Corollary~\ref{cor1} and
Lemma~\ref{propPPC} hold, then $e_1(t)$ converges with rate at
least $l_{u_0}$ to:
\[
\mathcal E_\infty
\coloneqq
\{e_1\in\mathbb R:|e_1|\leq e_\infty\}.
\]
\end{enumerate}
\end{thm}
\begin{proof} 
Let us define $u_d^o\coloneqq u_d(\hat s,\rho_u),$ $u_r^o\coloneqq u_r(\hat s,\rho_u,\rho_r,u),$ $e_u^o\coloneqq
u-\operatorname{sat}_{\bar u}(u_d^o)$ and $\xi_r^o\coloneqq\frac{e_u^o}{\rho_r}.$ The initialization conditions imply
$|\xi_o(0)|<1$, $|\xi_r^o(0)|<1$, and
$\rho_u(0),\rho_r(0)>0$. The proof follows a separation principle by first establishing fast convergence of the nonlinear observer dynamics and then showing that the resulting closed-loop dynamics remain arbitrarily close to those of the state-feedback. By selecting $r_\infty$ sufficiently small, the observer subsystem converges faster than the state-feedback controller dynamics. As a consequence of this time-scale separation, the system state experiences only small variations during the observer transient. Once the PPO enters its steady-state, the output-feedback closed-loop system evolves in a neighborhood of the corresponding state-feedback trajectory, thereby recovering the performance of the state-feedback design. In what follows, we will omit the time argument in the closed-loop signals unless it is necessary, for clarity in the proof. Thence, we express the observer dynamics in terms of scaled estimation errors:
\begin{equation}\label{scaling}
    \begin{split}
        &\eta_1 (t) \coloneqq \frac{T \left( \frac{e_1(t) - \hat{e}_1(t)}{r(t)} \right)r(t)}{r_{\infty}^{n-1}}  \\ 
        &\eta_i (t) \coloneqq \frac{e_i(t) - \hat{e}_i(t)}{r_{\infty}^{n-i}},~ i=2,...,n.
    \end{split}
\end{equation}
It follows from \eqref{scaling} that
$\hat e=e-\mathcal D(r_\infty,r,\eta)$, where $\eta \coloneqq [\eta_1,\dots,\eta_n]^T$ and:
\[
\mathcal D(r_\infty,r,\eta)\coloneqq
\begin{bmatrix}
rT^{-1}\!\left(\dfrac{r_\infty^{n-1}\eta_1}{r}\right)\\
r_\infty^{n-2}\eta_2\\
\vdots\\
r_\infty\eta_{n-1}\\
\eta_n
\end{bmatrix},
~~
\xi_o\coloneqq
T^{-1}\!\left(\frac{r_\infty^{n-1}\eta_1}{r}\right).
\]
Subsequently, the closed-loop system can be written in the singularly perturbed form: 
\begin{align}
        \dot{s}(e) &= w^TAe + w^TB\Delta(e + x_d,z,d,\dot{x}_d,u)\\ 
        \dot{\rho}_u  & = \alpha_{\rho_u}(\hat{s}(e - \mathcal D(r_\infty,r,\eta)),\rho_u)  \label{eq:outadaptup} \\
            \dot{\rho}_r  & = \alpha_{\rho_r}(\hat{s}(e - \mathcal D(r_\infty,r,\eta)),\rho_u,\rho_r,u)  \label{eq:outadaptrp}  \\ 
             \dot{u} & = \alpha_{u}(\hat{s}(e - \mathcal D(r_\infty,r,\eta)),\rho_u,\rho_r,u) \label{eq:outcontrolp}  \\
              r_{\infty} \dot{\eta} & = A_o(J(\xi_o))\eta + r_\infty C \theta(\eta_1,\xi_o,r,\dot{r})  \label{eq:eta} \\& \hspace{1.35em} + r_{\infty}B\Delta(e + x_d,z,d,\dot{x}_d,u) \nonumber  
\end{align}
with:
\begin{align*}
& \Delta  \coloneqq f(e + x_d,z,d) -B^T \dot{x}_d  +  g(e + x_d,z,d)u   \\
& A_o(J)\coloneqq
\begin{bmatrix}
-J\gamma_1 & J & 0 & \cdots & 0\\
-\gamma_2 & 0 & 1 & \cdots & 0\\
\vdots & \vdots & \ddots & \ddots & \vdots\\
-\gamma_{n-1} & 0 & \cdots & 0 & 1\\
-\gamma_n & 0 & \cdots & 0 & 0
\end{bmatrix}\\
&\theta (\eta_1,\xi_o,r,\dot{r}) \coloneqq  \frac{\dot{r}}{r} \eta_1 \left( 1 - \frac{J(\xi_o)\xi_o}{T{(\xi_o)}}  \right) .
\end{align*}
Assume that there exist $\bar\xi_o\in(0,1)$,
$P=P^T>0$, and $q>0$ such that
\begin{equation}\label{obssector}
A_o(J)^TP+PA_o(J)\preceq-qI,
\qquad J\in[1,J(\bar\xi_o)].
\end{equation}
Since $A_o(J)$ is affine in $J$, it suffices to verify
\eqref{obssector} at $J=1$ and $J=J(\bar\xi_o)$. Moreover,
select $\mu_o\in(0,\bar\xi_o)$ such that
\begin{equation}\label{obsinit}
\frac{|e_1(0)-\hat e_1(0)|}{r(0)}\leq\mu_o,
\quad
\sqrt{\frac{P_{11}}{\lambda_{\min}(P)}}\,T(\mu_o)
<T(\bar\xi_o).
\end{equation}

Let $\chi\coloneqq[z^T,e^T,\rho_u,u]^T$, and let
$\mathcal K_0$ denote its compact initialization set. Since the
initial performance inequalities hold strictly, choose a compact
neighborhood $\mathcal K_1$ such that:
\[
\mathcal K_0\subset\operatorname{int}\mathcal K_1,
\qquad
|s(e)|<\Bar\xi_u\rho_u,
\qquad
\rho_u<\Bar\rho_u
\quad\text{on }\mathcal K_1.
\]
The normalized saturation of $\hat s/\rho_u$, the invariance of
$\mathcal U$, and $|\dot u|\leq\bar r$ imply that the dynamics
of $\chi$ are uniformly bounded on $\mathcal K_1$. Therefore,
there exist $\tau_1>0$ and $\bar\Delta>0$, independent of
$r_\infty$, such that
\[
\chi(t)\in\mathcal K_1,\qquad
|\Delta(t)|\leq\bar\Delta,\qquad t\in[0,\tau_1].
\]

We now establish an observer estimate valid throughout the
transient. Define
\[
b_o\coloneqq
\max_{|\xi|\leq\bar\xi_o}
\left|1-\frac{J(\xi)\xi}{T(\xi)}\right|,
\]
using its continuous value at $\xi=0$. Since
$|\dot r|/r\leq l_o$, the Lyapunov function
$W(\eta)=\eta^TP\eta$ satisfies, as long as
$|\xi_o|<\bar\xi_o$:
\[
\dot W\leq
-\left(\frac{q}{r_\infty}-2\|PC\|l_ob_o\right)\|\eta\|^2
+2\|PB\|\bar\Delta\|\eta\|.
\]
For sufficiently small $r_\infty$, there exist constants
$\alpha,c>0$, independent of $r_\infty$, such that
\begin{equation}\label{Westimate}
W(t)\leq
\max\left\{
W(0)e^{-\alpha t/r_\infty},cr_\infty^2
\right\}.
\end{equation}

To close the sector argument, define
\[
v_0(r_\infty)\coloneqq r_\infty^{n-1}\eta(0)
=
\begin{bmatrix}
r(0)T(\xi_o(0))\\
r_\infty\tilde e_2(0)\\
\vdots\\
r_\infty^{n-1}\tilde e_n(0)
\end{bmatrix}.
\]
Using $r(t)\geq r(0)e^{-l_ot}$ and \eqref{Westimate} gives:
\[
|T(\xi_o(t))|
\leq
\max\left\{
b_0(r_\infty)
e^{-\left(\frac{\alpha}{2r_\infty}-l_o\right)t},
\sqrt{\frac{c}{\lambda_{\min}(P)}}r_\infty^{n-1}
\right\}
\]
where
\[
b_0(r_\infty)\coloneqq
\frac{\sqrt{v_0(r_\infty)^TPv_0(r_\infty)}}
{r(0)\sqrt{\lambda_{\min}(P)}}.
\]
The compact initialization sets and \eqref{obsinit} imply
$b_0(r_\infty)<T(\bar\xi_o)$ for all sufficiently small
$r_\infty$. Select $r_\infty$ further such that:
\[
\frac{\alpha}{2r_\infty}>l_o,
\qquad
\sqrt{\frac{c}{\lambda_{\min}(P)}}r_\infty^{n-1}
<T(\bar\xi_o).
\]
A first-exit argument then gives
$|\xi_o(t)|<\bar\xi_o$ throughout the observer transient,
thereby validating \eqref{obssector}. Moreover, once $\eta$
enters $\mathcal H\coloneqq
\{\eta:W(\eta)\leq cr_\infty^2\}$ the second inequality guarantees that the sector condition
continues to hold.

Furthermore, compact initialization gives
$W(0)\leq c_0/r_\infty^{2(n-1)}$ for some $c_0>0$.
Thus, $\eta$ enters $\mathcal H$ no later than
\[
\tau(r_\infty)
=
\frac{r_\infty}{\alpha}
\ln\left(\frac{c_0}{cr_\infty^{2n}}\right)
\]
and $\tau(r_\infty)\to0$ as $r_\infty\to0$. Therefore,
$\tau(r_\infty)<\tau_1$ for all sufficiently small
$r_\infty$.

Since $\eta\in\mathcal H$, it holds $\|\mathcal D(r_\infty,r,\eta)\|\leq L_Dr_\infty$ for some $L_D>0$ independent of $r_\infty$.

Next, let $\hat\xi_u\coloneqq\hat s/\rho_u$ and define:
\[
M_u^o\coloneqq
\max_{|\xi|\leq1-\xi_s}
\frac{
\operatorname{sat}_{\Bar u}(-k_u\mathcal T(\xi))
+k_u\mathcal T(\xi)}
{\xi}
\]
where the quotient is continuously extended by zero at
$\xi=0$. Assuming that $\Bar r>\beta$, define also $\xi_{r,a}\coloneqq
 T^{-1}\!\left(
\frac{\Bar r-\beta}{k_r}
\right)>0.$ Let $\Bar\rho_u$ and $\Bar\rho_r$ be auxiliary constants
satisfying the state-feedback feasibility condition
\eqref{eq:actuator-feasibility}. We impose the additional
output-feedback compatibility condition:
\begin{equation}\label{eq:of-feasibility}
M_u^o
<
l_{u_0}(\Bar\rho_u-\rho_u^\infty),
~~
\max\left\{
\rho_r(0),\rho_r^\infty,
\frac{2\Bar u}{\xi_{r,a}}
\right\}
<\Bar\rho_r.
\end{equation}
Consequently, constants $\Bar\rho_u^+$ and
$\Bar\rho_r^+$ can be selected such that $\max\left\{
\rho_u(0),
\rho_u^\infty+\frac{M_u^o}{l_{u_0}}
\right\}
<\Bar\rho_u^+<\Bar\rho_u$
and $\max\left\{
\rho_r(0),\rho_r^\infty,
\frac{2\Bar u}{\xi_{r,a}}
\right\}
<\Bar\rho_r^+<\Bar\rho_r.$ The output-feedback adaptive law gives $\dot\rho_u
\leq
M_u^o-l_{u_0}(\rho_u-\rho_u^\infty).$
Therefore, at $\rho_u=\Bar\rho_u^+$ it holds $\dot\rho_u
\leq
M_u^o-l_{u_0}
(\Bar\rho_u^+-\rho_u^\infty)<0.$
At $\rho_u=\rho_u^\infty$, the adaptive deficiency term is
nonnegative. Hence, $\rho_u^\infty
\leq\rho_u(t)\leq\Bar\rho_u^+.$

We next verify that the unsaturated transformation
$\xi_r^o=e_u^o/\rho_r$ remains well defined. Since
$u(0)\in\mathcal U$, the set
$\mathcal U=[-\Bar u,\Bar u]$ is positively invariant.
Indeed, at $u=\Bar u$, one has $\xi_r^o\geq0$ and hence
$\dot u\leq0$, whereas at $u=-\Bar u$, one has
$\xi_r^o\leq0$ and hence $\dot u\geq0$. Therefore, $|e_u^o|\leq2\Bar u.$

At $\rho_r=\Bar\rho_r^+$, one has $|\xi_r^o|
\leq\frac{2\Bar u}{\Bar\rho_r^+}
<\xi_{r,a}.$
Hence, the rate saturation is inactive and $\dot\rho_r
=
-l_r(\rho_r-\rho_r^\infty)<0.$ At $\rho_r=\rho_r^\infty$, the adaptive deficiency term is
nonnegative. It follows that $\rho_r^\infty
\leq\rho_r(t)\leq\Bar\rho_r^+.$ On the maximal interval where $|\xi_r^o|<1$, its dynamics
can be written as $\dot\xi_r^o
=
\frac{1}{\rho_r}
\left(
q_r^o-k_r\mathcal T(\xi_r^o)
\right)$, with $q_r^o\coloneqq
-\frac{d}{dt}
\operatorname{sat}_{\Bar u}(u_d^o)
+\xi_r^ol_r(\rho_r-\rho_r^\infty).$
The preceding compact bounds and the observer estimate imply
that $|q_r^o(t)|\leq Q_r^o$ on the maximal interval for some $Q_r^o>0$. Choose
$\Bar\xi_r^o\in(|\xi_r^o(0)|,1)$ sufficiently close to one
such that $k_r\mathcal T(\Bar\xi_r^o)>Q_r^o.$ Then, at $\xi_r^o=\Bar\xi_r^o$, $\dot\xi_r^o
\leq
\frac{Q_r^o-k_r\mathcal T(\Bar\xi_r^o)}{\rho_r}<0$ whereas at $\xi_r^o=-\Bar\xi_r^o$, $\dot\xi_r^o
\geq
\frac{-Q_r^o+k_r\mathcal T(\Bar\xi_r^o)}{\rho_r}>0.$
Thus, a first-exit argument gives $|\xi_r^o(t)|\leq\Bar\xi_r^o<1$, for all $t\in[0,\tau_{\max}).$

Let $\bar\xi_u<1$ denote the compact normalized-error bound
obtained in Theorem~\ref{lemma1}. The stable filter defining
$s$, the compact reference set, and Assumption~\ref{ass2}
provide compact sets $\Omega_e$ and $\Omega_z$ whenever
$|s|\leq\bar\xi_u\rho_u$ and
$\rho_u\leq\bar\rho_u^+$. Define $\psi\coloneqq[z^T,e^T,\rho_u,u]^T$
and the closed compact set:
\[
\Psi\coloneqq
\left\{
\psi:
\begin{array}{l}
z\in\Omega_z,\quad e\in\Omega_e,\\
|s(e)|\leq\bar\xi_u\rho_u,\\
\rho_u^\infty\leq\rho_u\leq\bar\rho_u^+,\\
u\in[-\bar u,\bar u]
\end{array}
\right\}.
\]
The sets $\Omega_e$, $\Omega_z$, and $\Psi$ are selected so
that the portion of $\mathcal K_1$ reached before
$\tau(r_\infty)$ lies in $\operatorname{int}\Psi$.

Denote by $F_s$ and $F_o$ the state and output-feedback
vector fields of $\psi$, respectively, with $\rho_r$ regarded
as a parameter in the compact interval
$[\rho_r^\infty,\bar\rho_r^+]$. Since
$\eta\in\mathcal H$ it holds $\|\mathcal D(r_\infty,r,\eta)\|\leq L_Dr_\infty$, for some $L_D>0$ independent of $r_\infty$. The normalized saturation is inactive on the state-feedback
invariant set. Moreover, the applied rate-control mapping $\xi\mapsto
\operatorname{sat}_{\bar r}(-k_r\mathcal T(\xi))$
admits a Lipschitz-continuous extension to $[-1,1]$, since it
is constant near $\xi=\pm1$. Hence, the local Lipschitz
continuity of the remaining controller terms gives:
\begin{equation}\label{L}
\left\|
F_o\bigl(\psi,\mathcal D(r_\infty,r,\eta)\bigr)
-F_s(\psi)
\right\|
\leq Lr_\infty
\end{equation}
on $\Psi\times[\rho_r^\infty,\bar\rho_r^+]\times\mathcal H$ for some $L>0$ independent of $r_\infty$.

Define $V_u =
\frac{1}{2}
T^2\!\left(\frac{s(e)}{\rho_u}\right)$ and the performance-boundary component $\partial_s\Psi\coloneqq
\left\{
\psi\in\Psi:
|s(e)|=\Bar\xi_u\rho_u
\right\}.$ Let $\dot V_u^s$ and $\dot V_u^o$ denote the derivatives
of the same function $V_u$ along the state-feedback and
output-feedback vector fields, respectively, i.e., $\dot V_u^s
=
\nabla_\psi V_u(\psi)^\top F_s(\psi)$ and $\dot V_u^o
=
\nabla_\psi V_u(\psi)^\top
F_o\bigl(\psi,\mathcal D(r_\infty,r,\eta)\bigr)$. Since $\Bar\rho_u^+<\Bar\rho_u$ and $\Bar\rho_r^+<\Bar\rho_r$, the set $\Psi\times
[\rho_r^\infty,\Bar\rho_r^+]$ is contained in the candidate compact set used in
Theorem~\ref{lemma1}. Therefore, the bound $|h|\leq H_u$
derived in the state-feedback proof remains valid on this set.
For clarity, denote the derivatives of $V_u$ along the
state and output-feedback vector fields by $\dot V_u^s$ and $\dot V_u^o$, respectively. On $\partial_s\Psi$, the strict estimate
from Theorem~\ref{lemma1} gives $\dot V_u^s
\leq-m_u^s<0.$ Since $V_u^s$ is continuously differentiable on $\Psi$, define $c_V\coloneqq
\max \limits _{\psi\in\Psi}
\|\nabla_\psi V_u(\psi)\|<\infty.$ Using \eqref{L}, we obtain, on $\partial_s\Psi$:
\begin{align*}
\dot V_u^o
 \leq 
\dot V_u^s
+
\nabla_\psi V_u(\psi)^TLr_\infty \leq
-m_u^s+c_VLr_\infty.
\end{align*}
Hence, if $r_\infty<
\frac{m_u^s}{2c_VL}$, then $\dot V_u^o \leq-\frac{m_u^s}{2}<0$ on $\partial_s\Psi$. Therefore, the output-feedback vector field points strictly
inward on the performance-boundary component
$\partial_s\Psi$, and the trajectory cannot leave $\Psi$
through this boundary. The stable filter defining $s$ and
Assumption~\ref{ass2} then retain $e$ and $z$ in
$\Omega_e$ and $\Omega_z$, respectively. Finally, the same
Lyapunov estimate for $W$, with $\Delta$ bounded on this
compact set, establishes the positive invariance of
$\mathcal H$.

Let $r^*>0$ be sufficiently small such that, for every
$r_\infty\in(0,r^*)$, all the preceding smallness
requirements hold, i.e., $b_0(r_\infty)<T(\Bar\xi_o)
$, $\frac{\alpha}{2r_\infty}>l_o$, $\sqrt{\frac{c}{\lambda_{\min}(P)}}
r_\infty^{n-1}<T(\Bar\xi_o)$, $\tau(r_\infty)<\tau_1$ and $r_\infty<
\frac{m_u^s}{2c_VL}$. Then, for every
$r_\infty\in(0,r^*)$, the slow variables remain in $\mathcal K_1$ until
$\eta$ enters $\mathcal H$. Subsequently, the trajectory remains in the compact set $\left\{
(\psi,\rho_r,\eta)\in
\Psi\times[\rho_r^\infty,\Bar\rho_r^+]\times\mathcal H:
|\xi_r^o|\leq\Bar\xi_r^o
\right\}$. The continuation theorem therefore gives
$\tau_{\max}=\infty$, and all closed-loop signals remain
bounded. The definition of $\mathcal K_1$ guarantees $|s(e(t))|<\Bar\xi_u\rho_u(t)$ during the observer transient, while the positive invariance of $\Psi$ guarantees the same bound afterward. Consequently, $|s(e(t))|
\leq\Bar\xi_u\rho_u(t)
<\rho_u(t)$ for all $t\geq0$, which proves claim~2.

Finally, if the conditions of Corollary~\ref{cor1} hold and
$l_{u_0}<\min_{j=1,\ldots,n-1}\lambda_j$, then
Lemma~\ref{propPPC} gives $|e_i(t)|
\leq
\bar e_i e^{-l_{u_0}(t-\tau_0)}
+
\frac{2^{i-1}\rho_u^\infty}
{\delta\prod_{j=1}^{n-i}\lambda_j}$ for $i=1,\ldots,n$.
Recalling that $\rho_u^\infty
=\delta e_\infty\prod_{j=1}^{n-1}\lambda_j$,
we conclude that the output tracking error converges exponentially to
$\mathcal E_\infty$, completing the
proof.
\end{proof}
\section{Simulation Results}\label{simsec}
\subsection{Simulation I: Application to a Flexible Joint Robot}
To validate the efficacy of the proposed controller in systems with high relative degree, we consider a single-link flexible-joint robot arm, modeled as in \cite{mousavi}. The dynamics of the system is transformed into the canonical form \eqref{eq:system}, with \(n = 4\), where the input gain is given by:
\begin{align*}
    g=\frac{k_F}{J_1 J_2 N}
\end{align*}
and the nonlinear function \(f(x)\) is expressed as:  
\begin{align*}
     f(x) = 
    &   \frac{k_F}{J_2 N} \left( \frac{k_F}{J_1 N} - \frac{k_F}{J_2} \right)x_1   \\
    & - \left( \frac{k_F}{J_2 N} \left(1 + \frac{F_1 F_2 N}{J_1 k_F}\right) + \frac{k_F}{J_2} \right)x_3  \\
     & - \frac{k_F}{J_2 N} \left( \frac{F_2}{J_2} + \frac{F_1 N}{J_1} \right)x_2  - \left( \frac{F_1}{J_1} + \frac{F_2}{J_2} \right)x_4\\
    & 
    + \frac{mgd_F}{J_2}(x_3 \sin(x_1) + x_2^2 \cos(x_1)) \\
    & - \frac{k_F mgd_F}{J_2^2 N} \left(\cos(x_1) - \frac{F_1 J_2 N}{J_1 k_F} x_2 \sin(x_1)\right).
\end{align*}
 The physical parameters used in the simulation are \(F_1 = 0.1\), \(F_2 = 0.15\), \(J_1 = 0.15\), \(J_2 = 0.2\), \(k_F = 0.4\), \(N = 2\), \(m = 0.8\), \(g = 9.81\), and \(d_F = 0.6\) with an initial state $x(0) = [0.5,0,0,0]^T$. Note that the proposed controller relies solely on knowledge of the system's relative degree and the sign of the input gain, without requiring explicit information about the system nonlinearities.

The objective is to track the reference trajectory $y_r(t)= 1.5 \sin(t)$ using the output-feedback controller \eqref{PPOE}-\eqref{eq:outcontrol}. The PPO is implemented as:
\begin{equation*}
\begin{split}
    &\dot{\hat{e}}_1 = \hat{e}_{2} + \frac{\gamma_1 r(t)}{r_\infty}  T \left( \frac{e_1(t)-\hat{e}_1(t)}{r(t)} \right),~\hat{e}_1(0)=0 \\
    &\dot{\hat{e}}_2 = \hat{e}_{3} + \frac{\gamma_2 r(t)}{r_\infty^2}  T \left( \frac{e_1(t)-\hat{e}_1(t)}{r(t)} \right),~\hat{e}_2(0)=0 \\
    &\dot{\hat{e}}_3 = \hat{e}_{4} + \frac{\gamma_3 r(t)}{r_\infty^3}  T \left( \frac{e_1(t)-\hat{e}_1(t)}{r(t)} \right),~\hat{e}_3(0)=0 \\
    &\dot{\hat{e}}_4 =  \frac{\gamma_4 r(t)}{r^4_\infty}T\left( \frac{e_1(t)-\hat{e}_1(t)}{r(t)} \right) , ~\hat{e}_4(0)=0 \\ 
\end{split}
\end{equation*}
with $\gamma_1 =6,\gamma_2=11,\gamma_3=6,\gamma_4=1, r(0)=4,r_{\infty}= 0.002, l_o=10$. The tracking error surface is defined as:
\begin{align*}
    \hat{s} = \rho_u(t)\operatorname{sat}_{1-\xi_s}
\left(\frac{\delta(\hat{e}_4+\omega_3\hat{e}_3+\omega_2\hat{e}_2+\omega_1\hat{e}_1) }{\rho_u(t)}\right)
\end{align*}
with $\delta=0.2$, $\xi_s=0.01$ as well as  $\lambda_1=1.6$, $\lambda_2=1.875$,
and $\lambda_3=2$, yielding
$\omega_1=6$, $\omega_2=9.95$, and $\omega_3=5.475$.
Thus, $l_{u_0}=1.5<\min_i\lambda_i=1.6$, as required by
Lemma~\ref{propPPC}. The control law is obtained by:
\begin{align*}
    \dot{u}  = -\mathrm{sat}_{\Bar{r}}\left(k_r \mathcal{T} \left(  \frac{u + \mathrm{sat}_{\Bar{u}}\left( k_u \mathcal{T} \left(  \frac{\hat{s}}{\rho_u} \right)\right)}{\rho_r} \right)\right)  ,~u(0)=0
\end{align*}
with control gains $k_u=3,k_r=10$ and saturation levels $\Bar{u}=10,\Bar{r}=500$. The adaptive performance laws are given by:
\begin{align*}
    & \dot{\rho}_u  = \alpha_{\rho_u}(\hat{s},\rho_u)  ,~\rho_u(0)=30 \\
           & \dot{\rho}_r  = \alpha_{\rho_r}(\hat{s},\rho_u,\rho_r,u) ,~\rho_r(0)=5
\end{align*}
with $\alpha_{\rho_u}(\cdot),\alpha_{\rho_r}(\cdot)$ defined in \eqref{eq:adaptu} and \eqref{eq:adaptr}, respectively. The adaptive law parameters are set as $l_{u_1}=3,l_{u_0}=1.5,\rho_u^{\infty}=0.2,l_{r_1}=0,l_{r_0}=10,\rho_r^{\infty}=0.01$. Thus, the steady-state tracking error bound is $ e_\infty= \frac{\rho_{u}^{\infty}}{\delta \prod_{j=1}^{n-1} \lambda_{j}}  \approx 0.1667$. 

The results of the simulation are presented in Fig. \ref{f1}, which illustrates the performance of the proposed output-feedback controller. In particular, Fig. \ref{f1}(a) depicts the evolution of the tracking error $e_1$ along the adaptive performance funnel (APF), i.e., the time-varying envelope defined by \eqref{Sset}. Fig. \ref{f1}(b) displays the evolution of the error signal $e_u$, capturing the impact of rate constraints on the control input. Moreover, Fig.\ref{f1}(c) depicts the actual control input $u$, which adheres to amplitude and rate constraints, in comparison to the unconstrained ideal control input $u_d$ and its amplitude-limited counterpart $\mathrm{sat}_{\Bar{u}}(u_d)$. Finally, Fig. \ref{f1}(d) presents the observer estimation errors $\Tilde{e}_i = e_i - \hat{e}_i,~i=1,\dots,4$, highlighting the accuracy of the proposed observer, after a short transient period.

\begin{figure*}[thpb]
      \centering    \includegraphics[clip,trim={0.0cm 0.0cm 2.0cm 0.0cm},width=1\textwidth]{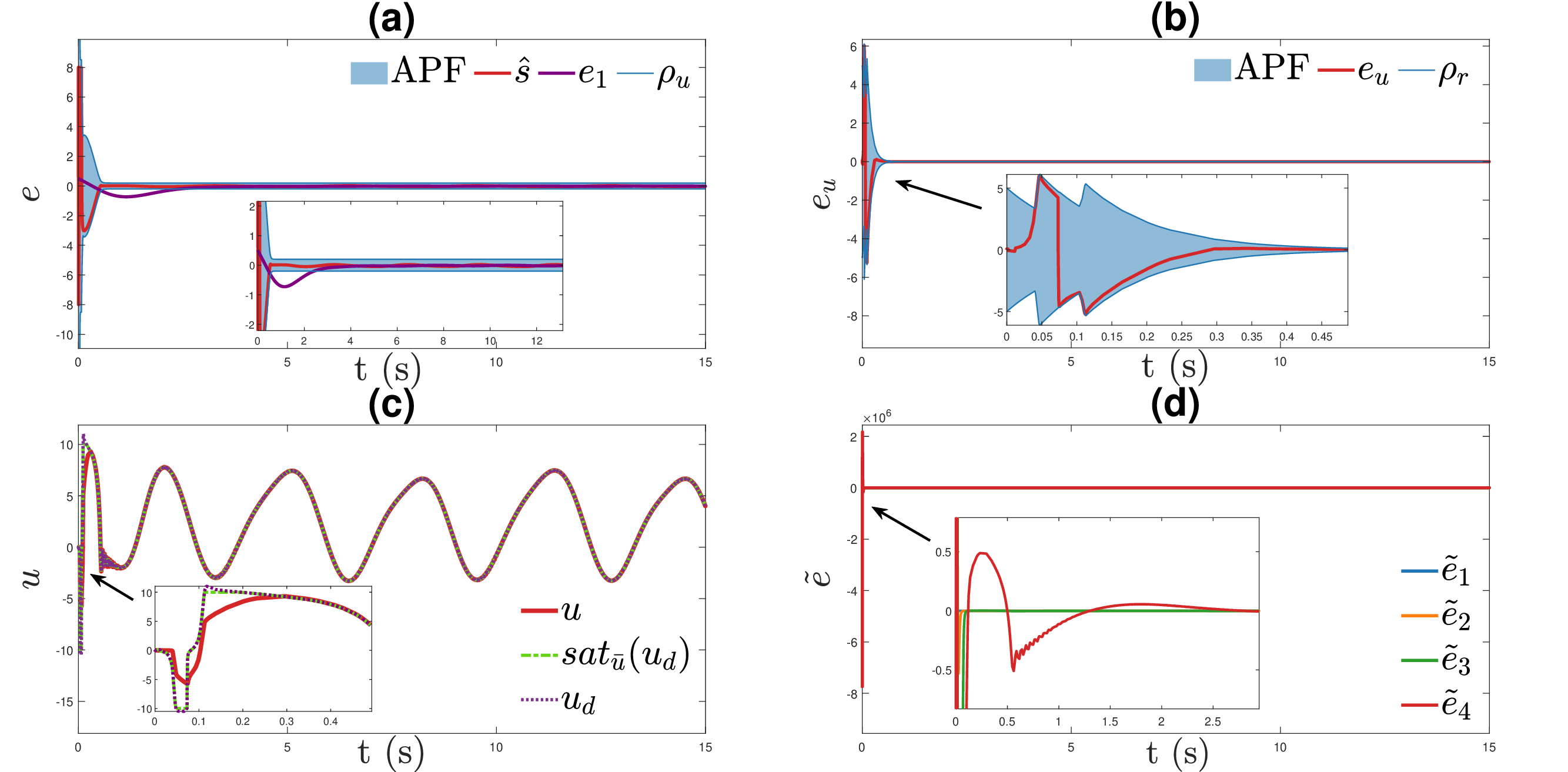}
   \caption{Simulation I: (a) Evolution of $e_1$ along with the adaptive performance funnel (APF) defined by $\rho_u$; (b) evolution of $e_u$ along with the corresponding APF defined by $\rho_r$; (c) evolution of the actual control input $u$ along with the unconstrained ideal control input $u_d$ and the amplitude saturated $\mathrm{sat}_{\Bar{u}}(u_d)$; (d) evolution of estimation errors $\Tilde{e}_i,~i=1,\dots,4$ of the PPO.}
      \label{f1}
\end{figure*}

\subsection{Simulation II: Application to wing rock motion control}
To further validate the proposed control scheme, we apply it to the wing-rock motion of a delta-wing aircraft. Wing rock is a nonlinear oscillatory rolling motion that, if left uncontrolled, can exhibit increasing amplitude over time. To account for state-dependent control gain variations and possible zero dynamics, we modify the mathematical model from \cite{wing1} for an $80^\circ$ delta-wing aircraft as follows:
\begin{align*}
    & \dot{z} = -2z +x_1^3 \\
    & \dot{x}_1  = x_2 \\
    & \dot{x}_2  = f(x,d) + (2+ \cos{(zx_1)})u \\ 
    & y= x_1
\end{align*}
with: 
\begin{align*}
  f(x,d) = &  - 0.008x_1   -0.03 x_2  - 0.4 \lvert x_2 \rvert x_2\\ &-0.01 x_1^3 - 0.06 x_1^2 x_2 + d
\end{align*}
where $x_1$ denotes the roll angle in rad, $u$ is the constrained control input and $d(t) = 0.25 \cos{(6t)}$ models the external disturbances. The reference trajectory for the roll angle is set to $y_r(t) = 0.5\sin(2t)$. 

The simulation study is conducted under two distinct scenarios to demonstrate: i) state-feedback performance recovery under different observer parameters and input constraints, compared to the method in \cite{dimanidis} and ii) the output-feedback control performance for varying initial conditions. 
For both controllers, the performance parameters are selected as
$\rho_u(0)=5$, $\rho_u^\infty=0.05$, and $l_{u_0}=1$.
For the proposed controller, $\rho_u(t)$ evolves according to
\eqref{eq:outadaptu}, whereas the corresponding nominal
exponential boundary is
$(\rho_u(0)-\rho_u^\infty)e^{-l_{u_0}t}+\rho_u^\infty$. The control gain $k_u=1$ and the parameters $\delta=1,\omega_1=2,\gamma_1=4,\gamma_2=2$ are set identical for both the proposed controller and the one presented in \cite{dimanidis}. The rest of the parameters for the proposed controller are selected as $\xi_s= 0.01,l_{u_1}=2,\rho_r(0)=5,\rho_r^{\infty}=0.001,l_{r_0}=1,l_{r_1}=2,\Bar{r}=    25,r(0)=2,l_o=2$.

\textbf{\textit{Scenario A1: Output-feedback recovery of state-feedback performance under varying observer gains.}}

In this simulation scenario, we showcase the ability of the proposed output feedback controller to recover the state-feedback performance as the observer parameter $r_{\infty}$ (or $\epsilon$ for the HGO in \cite{dimanidis}) decreases, for a fixed input saturation level $\Bar{u}=2$. As illustrated in Fig. \ref{f2} the proposed controller ensures accurate tracking in each case, while the controller of \cite{dimanidis} requires higher observer gains to effectively track the output as the estimation error is amplified at steady-state. Specifically, Fig. \ref{f2}(a),(b),(c) depict the evolution of $e_1$, $\Tilde{e}_2$ and $u$ under the proposed control scheme, demonstrating effective tracking and boundedness of all closed-loop signals. Fig. \ref{f2}(d),(e),(f) show the corresponding signals under the controller in \cite{dimanidis}, where small observer gains lead to increased estimation errors, ultimately leading to unbounded closed-loop signals. The standard HGO in \cite{dimanidis}, scales the observer gain up to $\frac{1}{\epsilon^2}$. Although the steady coefficient multiplying $T\left( \frac{e_1(t)-\hat{e}_1(t)}{r(t)} \right)$ is
$O(r_\infty^{-1})$ for the PPO, the local
measurement-error differential gain remains
$O(r_\infty^{-2})$ as discussed in Section \ref{gainAmp}.
\begin{figure*}[thpb]
      \centering    \includegraphics[clip,trim={3.0cm 0.2cm 3.5cm 0.4cm},width=1\textwidth]{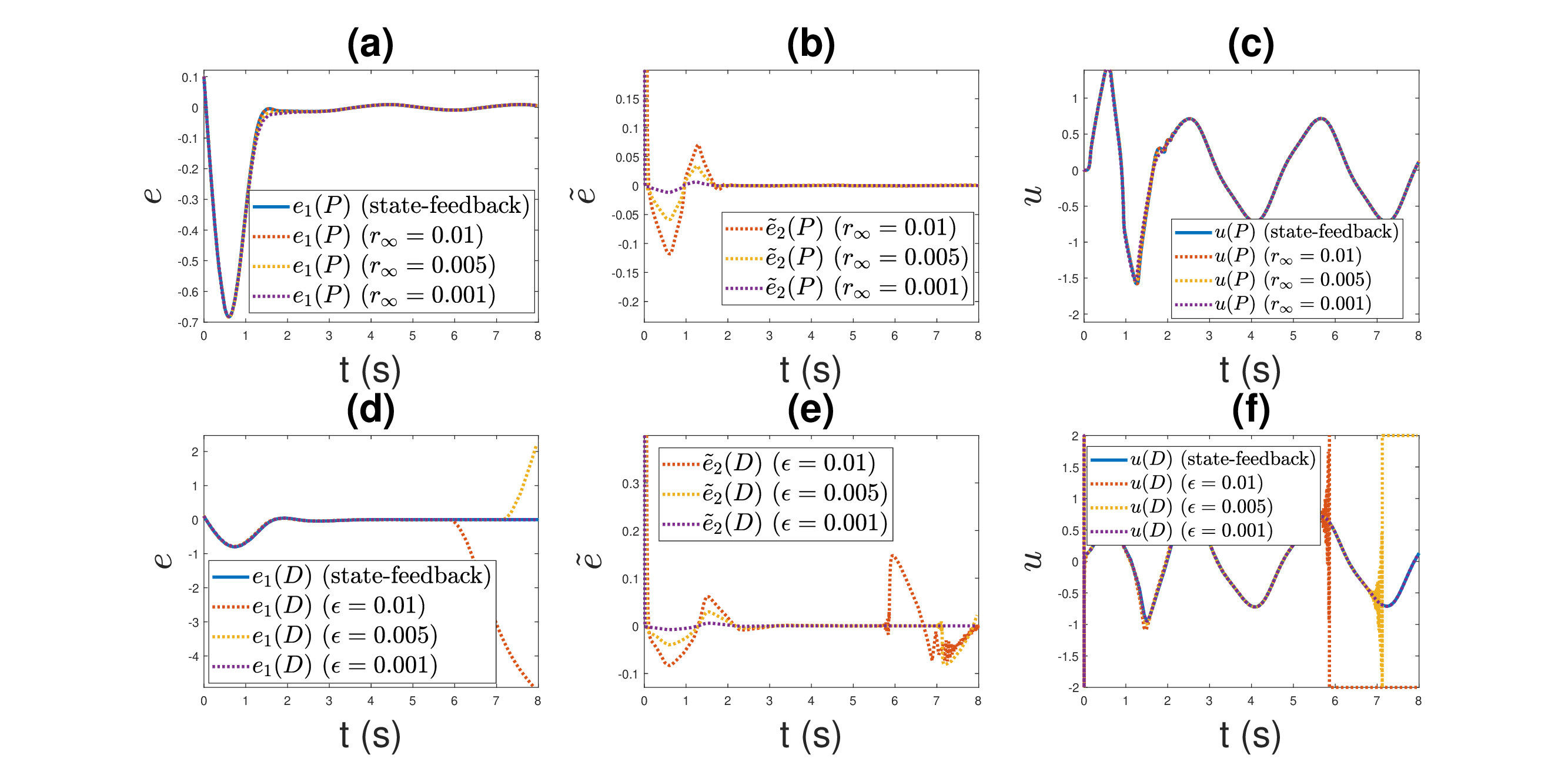}
   \caption{Simulation II - Scenario A1: $e_1(P),\Tilde{e}_2(P),u(P)$ refer to the signals under the proposed controller; $e_1(D),\Tilde{e}_2(D),u(D)$ refer to the signals under the controller in \cite{dimanidis}.}
      \label{f2}
\end{figure*}

\textbf{\textit{Scenario A2: Robustness comparison under heavier input saturation.}}

To further assess robustness, we reduce the input saturation level to $\Bar{u}=0.7$ to show the ability of the proposed controller to enlarge the region of attraction of the closed-loop system, compared to \cite{dimanidis}. As shown in Fig. \ref{f3}, the proposed output-feedback (with $r_{\infty}=0.01$) and state-feedback controllers successfully achieve trajectory tracking retaining the prescribed performance attributes at steady-state. In contrast, the state-feedback controller from \cite{dimanidis} exhibits an internal instability leading to unbounded system output. 
\begin{figure}[thpb]
      \centering    \includegraphics[clip,trim={0.5cm 0.2cm 1.5cm 0.6cm},width=1\columnwidth]{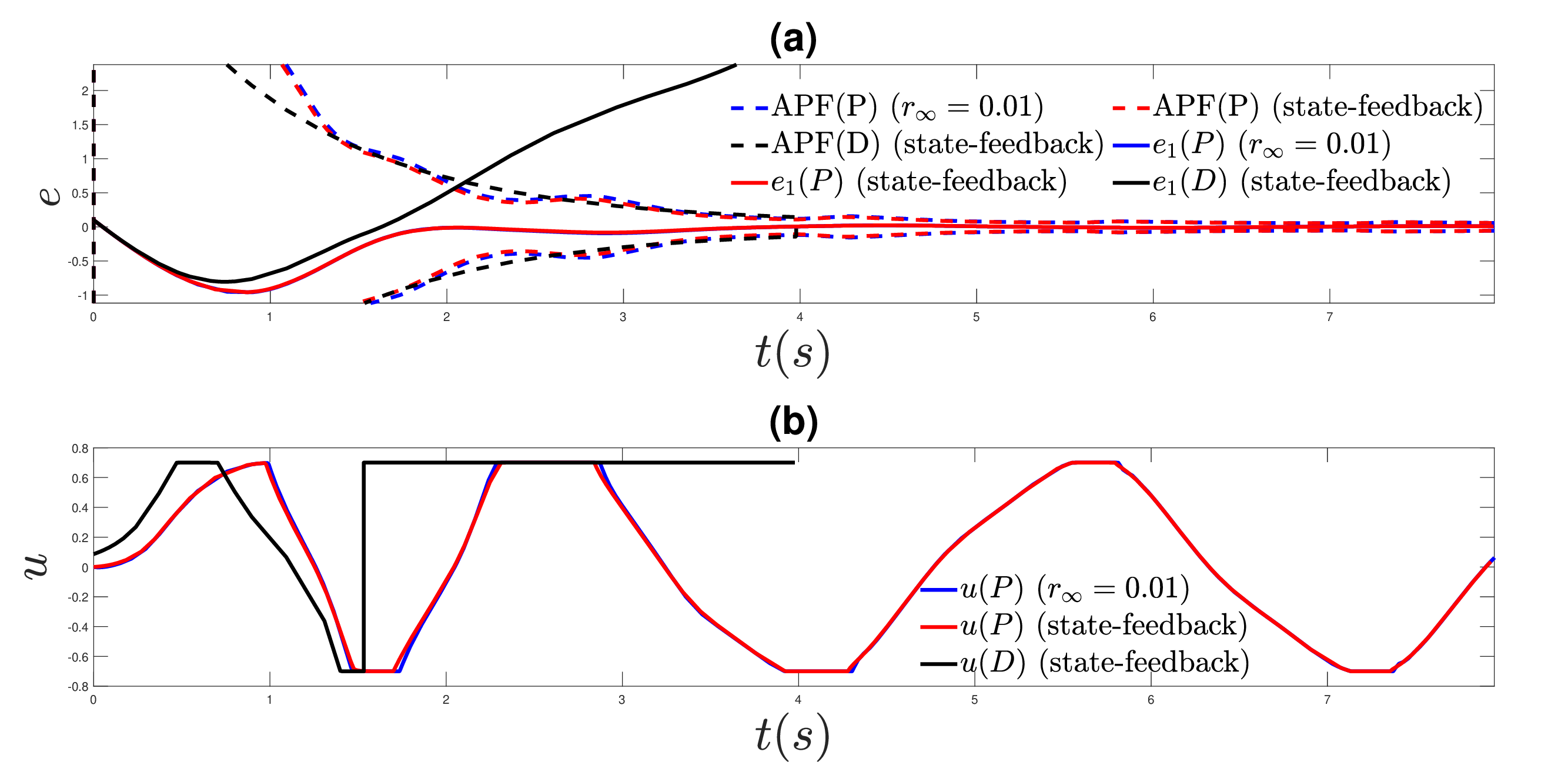}
   \caption{Simulation II - Scenario A2: (a) Evolution of e within the APF; (b) evolution of u, under the proposed output and state-feedback controller, compared to the state-feedback controller of \cite{dimanidis}.}
      \label{f3}
\end{figure}
As a consequence, the simulation results confirm the superior performance and robustness of the proposed adaptive output-feedback controller in recovering state-feedback performance even under limited actuation.

    \textbf{\textit{Scenario B: Output-feedback control for different initial conditions.}}
    In this simulation scenario, we test the performance of the proposed output-feedback controller to ensure trajectory tracking with adaptive performance for different initial conditions, in both position and velocity. Fig. \ref{f4}(a) illustrates the evolution of the tracking error $e_1$ for different initial conditions. Regardless of the initial state, the proposed controller successfully regulates the system toward the desired trajectory while respecting the adaptive performance specifications. Fig. \ref{f4}(b) shows the corresponding control input $u$. The input constraints are strictly satisfied across all cases, and the control signals converge to a common behavior after an initial transient phase. The transient response is influenced by the trade-off between input constraints and prescribed performance specifications, i.e. a more aggressive transient may be observed when higher control effort is available, allowing faster convergence. When input limitations are more restrictive, the system still adheres to the prescribed performance specifications at steady-state, but with a more gradual convergence.
    \begin{figure}[thpb]
          \centering    \includegraphics[clip,trim={3.0cm 0.2cm 3.5cm 0.6cm},width=1\columnwidth]{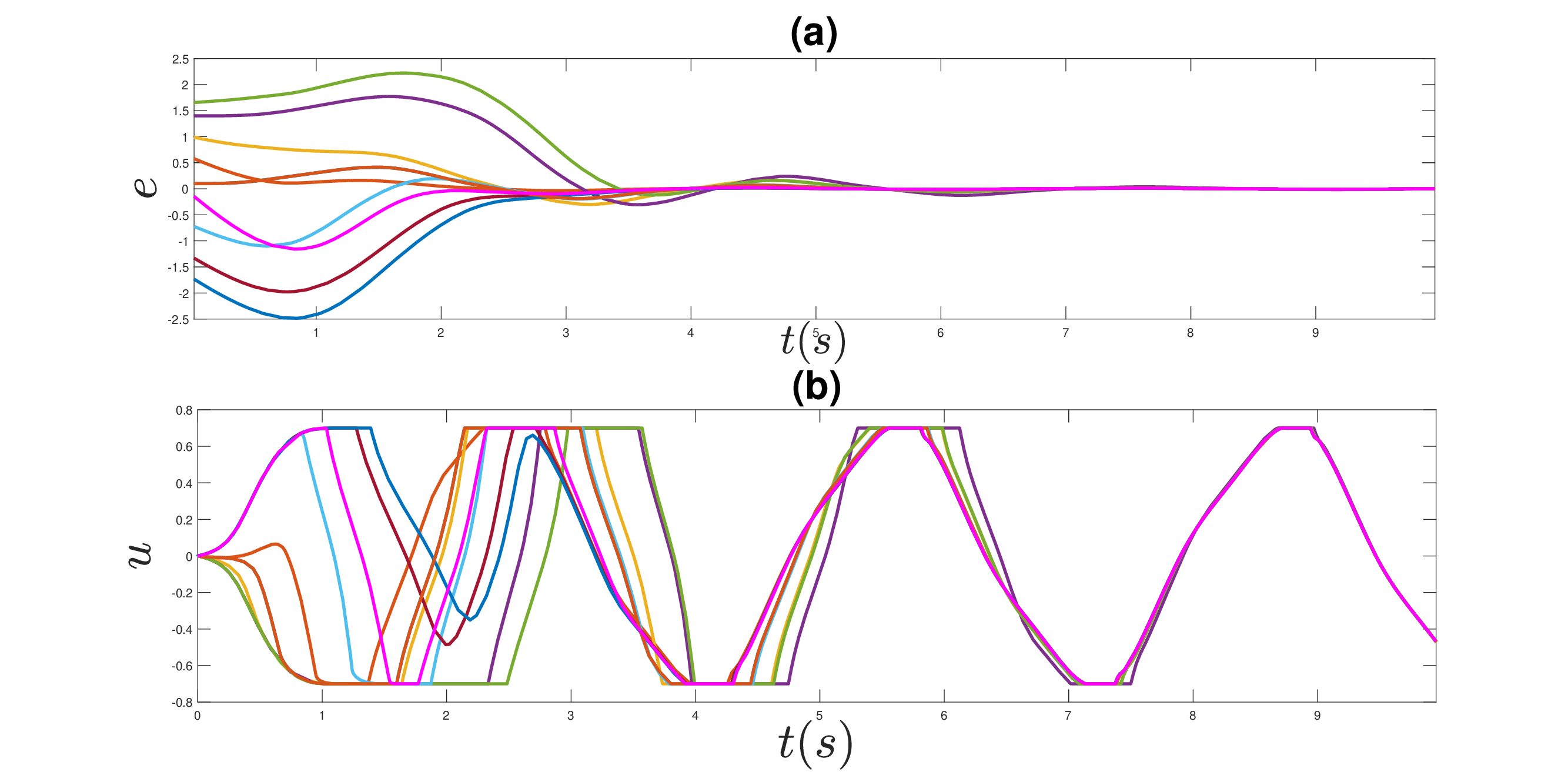}
       \caption{Simulation II - Scenario B:(a) Evolution of $e$ and; (b) evolution of $u$ under different initial conditions.}
          \label{f4}
    \end{figure}
\subsection{Simulation III: Comparative Study}
In this subsection, we conduct a comparative simulation study to evaluate the performance of the proposed control scheme against the approaches described in \cite{berger,funnel}. The study is based on a ship model described by \cite{c5}:
\begin{align}
\ddot{y} + \Phi \dot{y} + b_0 (M\dot{y} + Ly) = b_0 u
\end{align}
where \(y\) denotes the course angular velocity, \(u\) is the rudder angle, and \(b_0\), \(\Phi\), \(M\), and \(L\) are unknown constants associated with hydrodynamic coefficients and the ship's mass. For the simulation, we set the parameters as: \( \Phi = 0.2 \), \( b_0 = 1.85 \), \( M = 0.12 \), and \( L = 0.28 \) with the initial conditions $ y(0)=1.2,~\dot{y}(0)=1.5$. The control input is subject to amplitude constraints with a limit of $\Bar{u} = 2$. The desired trajectory is given by: $y_r(t) = 0.5\cos(t)$. The prescribed performance specifications are selected to enforce a steady-state tracking error of at most $e_{\infty}=0.1$ with a minimum convergence rate of $\exp{(-2t)}$. The parameters of the proposed controller are selected identical to Simulation II with $r_{\infty}=0.01$. For a fair comparison, the parameters of the competing controllers were fine-tuned as follows: $q_1=q_2=p_1=1,p_2=1/3,\Tilde{\gamma
}=1.25$ as defined in \cite{berger} and $M_k=2,\Bar{\sigma}_1=1.3,\Bar{\sigma}_2=0.3,\mu=0.1,k_2=1.1,\epsilon=0.01$ as defined in \cite{funnel}. For clarity, we denote the signals corresponding to the proposed scheme with (P), the controller in \cite{berger} with (B) and the controller in \cite{funnel} with (C). The simulation results are presented in Fig. \ref{f5}, illustrating: (a) the evolution of the tracking error $e_1$ for each control method and (b) the control input $u$ applied under each approach. The proposed controller achieves faster convergence and smaller steady-state errors compared to both benchmark controllers. The control effort remains comparable across all approaches. However, unlike the competing controllers, the proposed method does not require prior knowledge of the input gain or the higher derivatives of $y_r$. In contrast, the approach in \cite{funnel} assumes exact knowledge of $b_0$ is exactly known, while the method in \cite{berger} relies on the parameter $\Tilde{\gamma}$ being sufficiently close $b_0$ for effective control. 

\begin{figure}[thpb]
      \centering    \includegraphics[clip,trim={3.0cm 0.2cm 3.5cm 0.6cm},width=1\columnwidth]{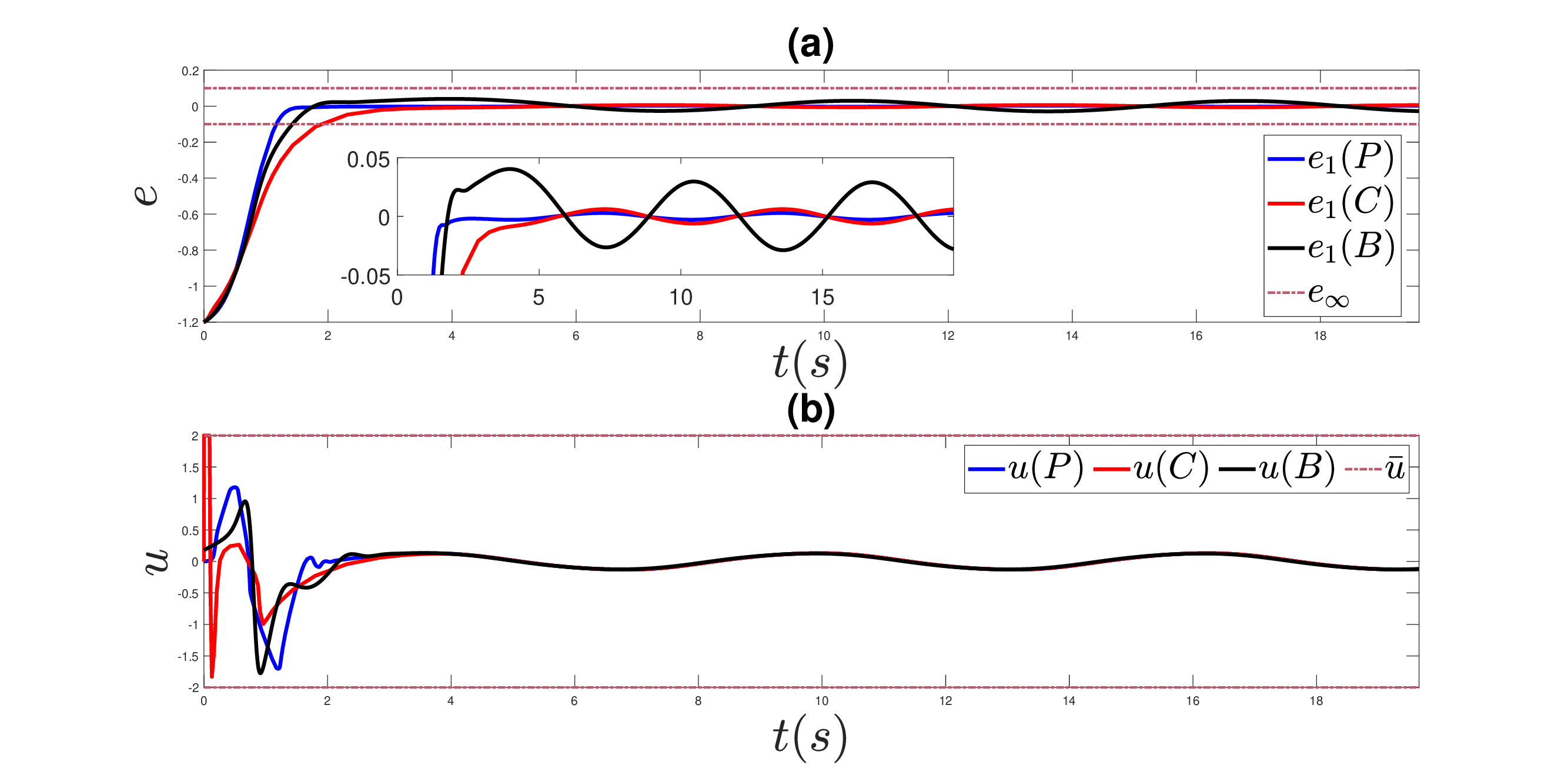}
   \caption{Simulation III: (a) Evolution of the tracking error $e_1$ under the proposed controller (P), \cite{berger} (B), and \cite{funnel} (C); (b) evolution of the control input $u$ for each method.}
      \label{f5}
\end{figure}

\textbf{\textit{Impact of Measurement Noise:}}
To further examine the robustness of the proposed control scheme, we extend the aforementioned study under the influence of measurement noise. The noise is generated by interpolating a vector of random numbers following normal distribution within the range $[-0.001, 0.001]$, with a sampling rate of $100~Hz$.

From Fig. \ref{fn}(a), it is evident that all three controllers successfully track the desired trajectory despite the presence of measurement noise, with the proposed controller exhibiting a slightly faster convergence rate. Examining Fig. \ref{fn}(b)-(d), we observe a significant difference in the smoothness of the control input. The proposed rate-limited controller (P) produces a significantly smoother control signal compared to the alternatives. In contrast, both the controllers from \cite{berger} and \cite{funnel} exhibit chattering effects induced by noise amplification. The estimation error $\Tilde{e}_2$, depicted in Fig. \ref{fn}(e)-(g), is comparable across all schemes. However, the controller from \cite{berger} demonstrates a lower estimation error.
Table~\ref{tab:metrics_all} summarizes standard performance indices for both noise-free and noisy scenarios. The settling time required for $|e_1(t)|$ to enter and remain within the prescribed bound $|e_\infty|$ is denoted by $T_s$, while $T$ denotes the duration of the simulation horizon used for computing the integral performance metrics.
\begin{table}[t]
\centering
\caption{Quantitative performance comparison}
\label{tab:metrics_all}
\setlength{\tabcolsep}{3pt}
\renewcommand{\arraystretch}{1.1}
\begin{tabular}{lcccccc}
\toprule
& \multicolumn{3}{c}{\textbf{Noise-free}}
& \multicolumn{3}{c}{\textbf{Noisy measurements}} \\
\cmidrule(lr){2-4}\cmidrule(lr){5-7}
\textbf{Metric}
& \textbf{(P)} & \textbf{(B)} & \textbf{(C)}
& \textbf{(P)} & \textbf{(B)} & \textbf{(C)} \\
\midrule
Settling time $T_s$ [s]
& \textbf{1.167} & 1.427 & 1.903
& \textbf{1.152} & 1.427 & 1.874 \\

$\frac{1}{T}\int_0^T |e_1(t)|\,dt$
& \textbf{0.046} & 0.067 & 0.062
& \textbf{0.091} & 0.115 & 0.118 \\

$\int_0^T |u(t)|\,dt$
& 2.864 & 2.730 & \textbf{2.557}
& 2.846 & \textbf{1.936} & 2.839 \\

$\int_0^T |\dot u(t)|\,dt$
& 7.540 & \textbf{6.893} & 11.085
& \textbf{25.715} & 64.794 & 334.784 \\
\bottomrule
\end{tabular}
\end{table} 
Consequently, the simulation results highlight the robustness of the proposed control approach. Despite operating without prior knowledge of system parameters, it achieves comparable or superior performance while mitigating the adverse effects of measurement noise, owing to the rate constraints. The smoother control action reduces the risk of actuator failure making it a more practical choice for real-world applications. However, further investigation is required to reduce the noise sensitivity of the PPO.

    \begin{figure*}[thpb]
          \centering    \includegraphics[clip,trim={2.0cm 0.2cm 2.5cm 0.6cm},width=1\textwidth]{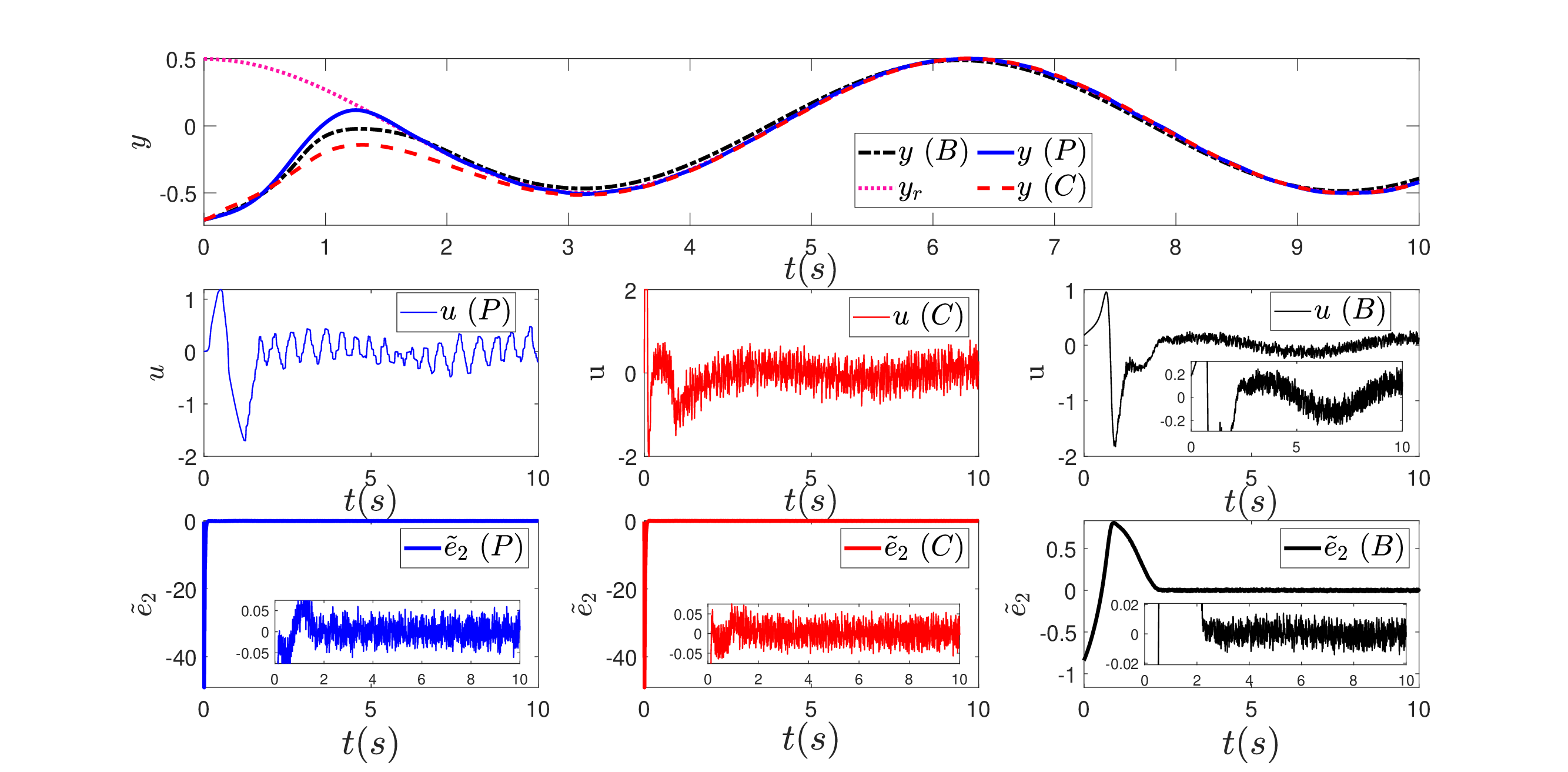}
       \caption{Simulation III - Impact of measurement noise: (a) Output response $y$ alongside the reference trajectory $y_r$; (b)-(d) control input $u$ corresponding to the proposed controller (P), \cite{berger} (B), and \cite{funnel} (C); (e)-(g) estimation error $\Tilde{e}_2$ for each method.}
          \label{fn}
    \end{figure*}

\section{Conclusions}\label{conclusio}
In this work, we developed a robust output-feedback control framework imposing adaptive output performance specifications suitable for uncertain high-order nonlinear systems in the normal form under amplitude and rate input constraints. A novel prescribed performance observer was developed to provide accurate state estimation with nonlinear gains depending on the estimation error. The integration of the PPO into an adaptive performance control scheme ensured bounded closed-loop signals and adaptive performance attributes based on the prescribed performance specifications and their conflict with the input limitations. The proposed framework is developed without leveraging explicit system dynamics or disturbance information, which enhances robustness and simplicity but introduce conservatism. Future work will focus on
incorporating partial model knowledge and data-driven approaches while
preserving robustness, as well as on reducing
the noise sensitivity of the PPO, as highlighted by the simulation results.

\bibliographystyle{IEEEtran}
\bibliography{IEEEabrv,references}
\begin{IEEEbiography}[{\includegraphics[width=1.0in,height=5.95in,clip,keepaspectratio]{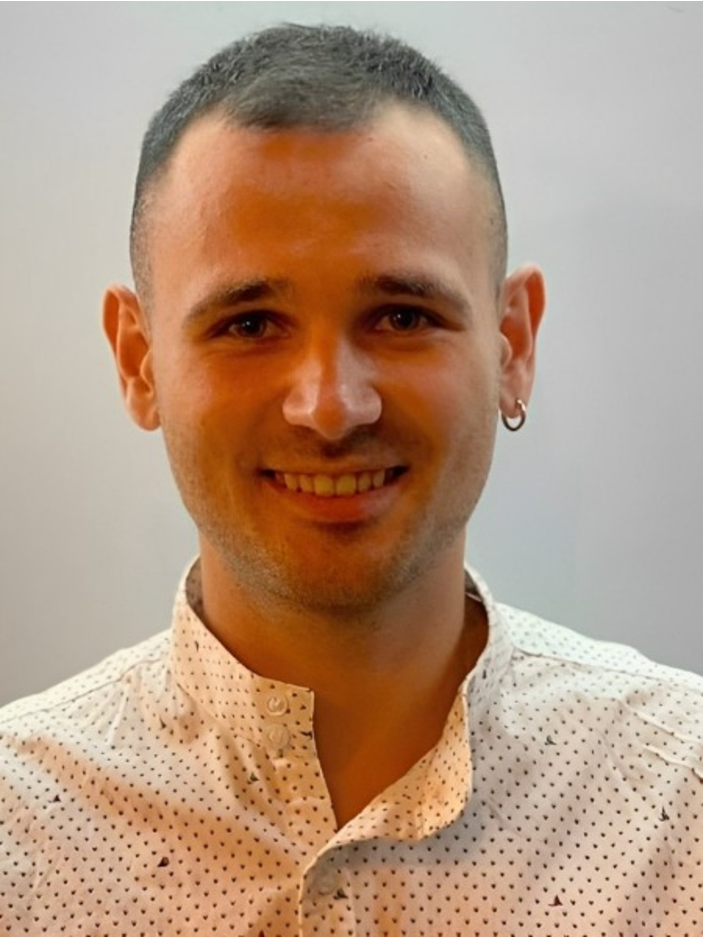}}]
{Panagiotis S. Trakas} was born in Tripoli, Greece, in 1998. He received the Diploma in Electrical and Computer Engineering from the Aristotle University of Thessaloniki, Thessaloniki, Greece, in 2021 and the Ph.D. degree in Electrical and Computer Engineering from the University of Patras, Patras, Greece in 2025. 

Between 2023 and 2024, he was a Visiting Scholar at the Division of Signals and Systems, Uppsala University, Uppsala, Sweden, and the Robot Perception and Learning (RPL) Lab at the University College London (UCL), London, U.K. He is currently a Postdoctoral Researcher at the Automatic Control Laboratory, Department of Information Technology and Electrical Engineering, ETH Zurich, Switzerland. His research interests include nonlinear robust adaptive control, autonomous systems, and multiagent systems.
\end{IEEEbiography}

\begin{IEEEbiography}[{\includegraphics[width=1in,height=1.25in,clip,keepaspectratio]{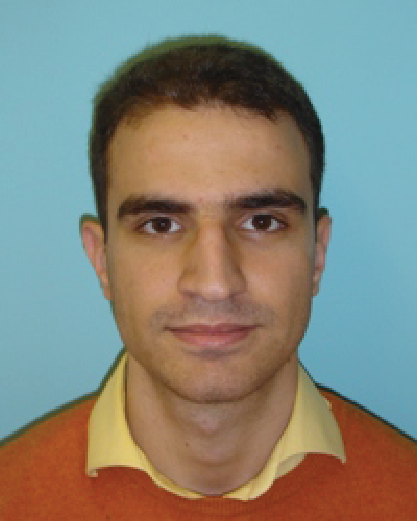}}]
{Charalampos P. Bechlioulis}(Senior Member, IEEE) was born in Arta, Greece, in 1983. He received the diploma degree (first in his class) in electrical and computer engineering (2006), the bachelor of science degree (second in his class) in mathematics (2011), and the Ph.D. degree in electrical and computer engineering (2011) from the Aristotle University of Thessaloniki, Thessaloniki, Greece.

He is currently a Professor with the Division of Systems and Control, Department of Electrical and Computer Engineering, University of Patras, Patras, Greece. He has authored more than 140 papers in scientific journals and conference proceedings and three book chapters. His research interests include nonlinear control with prescribed performance, system identification, control of robotic vehicles, and multiagent systems.
\end{IEEEbiography}
\end{document}

%% file: BD.tex
\tikzstyle{block} = [
  draw, fill=blue!5, rectangle,
  minimum height=3.6em,
  minimum width=3.6em,
  font=\large
]

\tikzstyle{block1} = [
  draw, fill=red!5, rectangle,
  minimum height=3.6em,
  minimum width=7.2em,
  font=\large
]

\tikzstyle{block2} = [
  draw, fill=green!5, rectangle,
  minimum height=3.6em,
  minimum width=7.2em,
  font=\large
]

\tikzstyle{sat} = [
  draw, fill=blue!5, rectangle,
  minimum height=2.6em,
  minimum width=3.1em,
  text width=1.4em,
  align=center,
  font=\large
]

\tikzstyle{sum}    = [draw, circle, node distance=1cm]
\tikzstyle{input}  = [coordinate]
\tikzstyle{output} = [coordinate]

\def\windup{
\tikz[remember picture,overlay]{
  \draw (-0.5,0) -- (0.5,0)  (0,-0.4)--(0,0.4);
  \draw (-0.53,-0.3)--(-0.3,-0.3) -- (0.3,0.3) --(0.53,0.3);
}}

\begin{tikzpicture}[auto, node distance=2cm, >=latex',scale=0.85,
  transform shape]

\node [input] (input) {};
\node [sum, right of=input] (sum) {};

\node [block1, right of=sum, node distance=2cm] (obs)
  {$\alpha_{\hat{e}}(e,\hat{e},r)$};

\node [block, right of=obs, node distance=3cm] (ud)
  {$u_d(\hat{s},\rho_u)$};

\node [sat, right of=ud, node distance=2.25cm] (u) {};
\node at (u) {\windup};

\node [block, right of=u, node distance=2.55cm] (ur)
  {$u_r(\hat{s},\rho_u,\rho_r,u)$};

\node [sat, right of=ur, node distance=2.7cm] (udot) {};
\node at (udot) {\windup};

\node [block, right of=udot, node distance=1.7cm] (ua) {$\int$};

\node [block2, right of=ua, node distance=2.5cm] (plant) {Plant};

\node [input, above of=plant] (d) {$d(t)$};
\node [output, right of=plant] (output) {};

\node [block, below of=ud, node distance=3cm, xshift=1.5cm] (rho)
  {$\alpha_{\rho_u}(\hat{s},\rho_u)$};

\node [sum, above of=rho, node distance=1.5cm] (sum1) {};

\node [block, below of=ur, node distance=3cm, xshift=1.8cm] (rhor)
  {$\alpha_{\rho_r}(\hat{s},\rho_u,\rho_r,u)$};

\node [sum, above of=rhor, node distance=1.5cm] (sum2) {};

\node [block2, below of=ur, node distance=5.2cm, xshift=-0.5cm] (sensors)
  {Measurements};

\draw[->] (input) -- node {$y_r -$} (sum);
\draw[->] (sum) -- node {$e_1$} (obs);
\draw[->] (obs) -- node {$\hat{e}$} (ud);

\draw[->] (ud) -- (u);
\draw[->] (u) -- (ur);
\draw[->] (ur) -- (udot);
\draw[->] (udot) -- node {$\dot{u}$} (ua);
\draw[->] (ua) -- node {$u$} (plant);

\draw[->] (plant) -- (output);
\draw[->] (d) -- node {$d$} (plant);

\draw[->] (u) |- node[pos=0.9] {$+$} (sum1);
\draw[->] (ud) -| node[pos=0.9] {$-$} (sum1);
\draw[->] (sum1) -- (rho);
\draw[->] (rho) -| node {$\rho_u$} (ud);

\draw[->] (udot) |- node[pos=0.9] {$+$} (sum2);
\draw[->] (ur) -| node[pos=0.9] {$-$} (sum2);
\draw[->] (sum2) -- (rhor);
\draw[->] (rhor) -| node {$(\rho_r,u)$} (ur);

\draw[->] (ua) |- node {$u$} (rhor);

\draw[->] (plant) |- (sensors);
\draw[->] (sensors) -| node[pos=0.95] {$+$} node[near end] {$y$} (sum);

\node (Y) [draw=blue, fit=(u)(ud)(rho)(ur)(udot)(rhor)(ua),
  inner sep=0.1cm, dashed, thick] {};
\node [yshift=2ex, blue] at (Y.north) {\textbf{Controller}};

\node (Y1) [draw=red, fit=(obs),
  inner sep=0.1cm, dashed, thick] {};
\node [yshift=2ex, red] at (Y1.north) {\textbf{Observer}};

\end{tikzpicture}